\documentclass[a4paper,twocolumn,11pt,accepted=2024-07-12]{quantumarticle}
\pdfoutput=1
\usepackage[utf8]{inputenc}
\usepackage[english]{babel}
\usepackage[T1]{fontenc}
\usepackage{amsmath}
\usepackage{hyperref}
\usepackage[numbers]{natbib}

\usepackage{amsfonts,color}
\usepackage{gnuplottex}
\usepackage{graphicx}
\usepackage{hyperref}
\hypersetup{pdftitle={Security of DM CVQKD}}
\usepackage{bbold}

\usepackage{tabularx}

\newcounter{protocol}
\newenvironment{protocol}[1]
  {\par\addvspace{\topsep}
   \noindent
   \tabularx{\linewidth}{@{} X @{}}
    \hrule
    \refstepcounter{protocol}\textbf{Protocol \theprotocol} #1 
    \hrule}
  { 
    \hrule
   \endtabularx
   \par\addvspace{\topsep}
   }

\newcounter{algorithm}
\newenvironment{algorithm}[1]
  {\par\addvspace{\topsep}
   \noindent
   \tabularx{\linewidth}{@{} X @{}}
    \hrule
    \refstepcounter{algorithm}\textbf{Algorithm \thealgorithm} #1 
    \hrule}
  { 
    \hrule
   \endtabularx
   \par\addvspace{\topsep}}

\newtheorem{theorem}{Theorem}

\newtheorem{lemma}{Lemma}

\newtheorem{proposition}{Proposition}

\newenvironment{proof}[1][Proof]{\noindent\textbf{#1.} }{\ \rule{0.5em}{0.5em}}

\def\freq{\operatorname{freq}}
\def\Tr{\operatorname{Tr}}
\def\Pr{\operatorname{Pr}}

\def\supp{\operatorname{supp}}

\def\>{\rangle}
\def\<{\langle}
\def\({\left(}
\def\){\right)}

\def\id{\operatorname{id}}
\let\oldemptyset\emptyset
\let\emptyset\varnothing

\newcommand{\ket}[1]{\left|{#1}\right\rangle}
\newcommand{\bra}[1]{\left\langle{#1}\right|}
\newcommand{\pro}[1]{\ket{#1}\bra{#1}}

\newcommand{\mc}[1]{\mathcal{#1}}

\newcommand{\R}{\mathbb{R}}

\begin{document}

\title{Security of discrete-modulated continuous-variable quantum key distribution}

\author{Stefan B{\"a}uml}
\affiliation{ICFO-Institut de Ciencies Fotoniques, The Barcelona Institute of Science and Technology, Av. Carl Friedrich Gauss 3, 08860 Castelldefels (Barcelona), Spain.}

\author{Carlos Pascual-Garc\'ia}
\affiliation{ICFO-Institut de Ciencies Fotoniques, The Barcelona Institute of Science and Technology, Av. Carl Friedrich Gauss 3, 08860 Castelldefels (Barcelona), Spain.}

\author{Victoria Wright}
\affiliation{ICFO-Institut de Ciencies Fotoniques, The Barcelona Institute of Science and Technology, Av. Carl Friedrich Gauss 3, 08860 Castelldefels (Barcelona), Spain.}

\author{Omar Fawzi}
\affiliation{Universit\'e de Lyon, Inria, ENS de Lyon, UCBL, LIP, F-69342, Lyon Cedex 07, France.}

\author{Antonio Ac\'in}
\affiliation{ICFO-Institut de Ciencies Fotoniques, The Barcelona Institute of Science and Technology, Av. Carl Friedrich Gauss 3, 08860 Castelldefels (Barcelona), Spain.}
\affiliation{ICREA - Instituci\'{o} Catalana de Recerca i Estudis Avan\c{c}ats, 08010 Barcelona, Spain.}

\begin{abstract}

Continuous variable quantum key distribution with discrete modulation has the potential to provide information-theoretic security using widely available optical elements and existing telecom infrastructure. While their implementation is significantly simpler than that for protocols based on Gaussian modulation, proving their finite-size security against coherent attacks poses a challenge. In this work we prove finite-size security against coherent attacks for a discrete-modulated quantum key distribution protocol involving four coherent states and heterodyne detection. To do so, and contrary to  most of the existing schemes, we first discretize all the continuous variables generated during the protocol. This allows us to use the entropy accumulation theorem, a tool that has previously been used in the setting of discrete variables, to construct the finite-size security proof. We then compute the corresponding finite-key rates through semi-definite programming and under a photon-number cutoff. Our analysis provides asymptotic rates in the range of $0.1-10^{-4}$ bits per round for distances up to hundred kilometres, while in the finite case and for realistic parameters, we get of the order of $10$ Gbits of secret key after $n\sim10^{11}$ rounds and distances of few tens of kilometres.

\end{abstract}

{\let\clearpage\relax\maketitle}

\tableofcontents

\section{Introduction}

Arguably one of the most technologically advanced applications of quantum information theory nowadays is quantum key distribution (QKD), which allows two honest parties, Alice and Bob, to obtain a cryptographic key, the security of which is guaranteed by the laws of quantum physics. Whereas QKD was originally conceived in a setting involving discrete variables \cite{BB84,E91,bennett1992quantum}, e.g. requiring the generation, or at least approximation, of single photon states, there exist a number of protocols based on continuous variable systems, such as squeezed or coherent states \cite{cerf2001quantum,grosshans2002continuous,weedbrook2004quantum,garcia2009continuous}. These protocols, known as continuous variable quantum key distribution (CVQKD), provide a number of advantages over discrete variable quantum key distribution (DVQKD) in terms of implementation using present-day telecom infrastructure.  

The security of DVQKD has been proven both in theory and in realistic implementations using diverse approaches, see for instance~\cite{shor2000simple,gottesman2004security,renner2005information,renner2008security,koashi2009simple}. Different security proofs have also been provided for CVQKD, many of which make use of a particular feature of the protocol, namely that the quantum states sent from Alice to Bob are chosen according to a Gaussian distribution. Such protocols are also known as Gaussian modulated CVQKD protocols. 
An important ingredient when proving security of Gaussian modulated CVQKD against collective attacks is the extremality of Gaussian states \cite{wolf2006extremality}. Gaussian extremality implies that, for a given covariance matrix of Alice and Bob's system, the maximum over the Holevo quantity in the Devetak-Winter formula for the key rate \cite{devetak2005distillation}, which involves an optimisation over Eve's full Fock space, is attained by the corresponding Gaussian state. Combining this with the fact that, in the case of Gaussian modulation, the covariance matrix of Alice and Bob's system can be directly computed from the observed statistics \cite{Grosshans2003Virtual}, security against collective attacks has been shown for Gaussian modulated coherent and squeezed states protocols involving both homodyne and heterodyne detection \cite{navascues2006optimality,garcia2006unconditional,leverrier2010simple,leverrier2015composable}.
Security against general attacks has been shown for protocols using coherent \cite{renner2009finetti,leverrier2013security,leverrier2015composable,leverrier2017security} as well as squeezed states \cite{furrer2012continuous,furrer2014reverse}. The main tools that have been used are the de Finetti Theorem \cite{renner2009finetti,leverrier2017security}, postselection techniques \cite{christandl2009postselection,leverrier2013security} and entropic uncertainty relations  \cite{furrer2012continuous,furrer2014reverse}. 

Unfortunately, the implementation of CVQKD protocols with Gaussian modulation turns out to be challenging because it is never achieved in practice \cite{lupo2020towards}, and is in fact often approximated by finite sets of states. A discrete modulation therefore significantly simplifies the preparation of states but also the error correction part, as much simpler reconciliation schemes can be used~\cite{leverrier2009unconditional}. Discrete-modulated protocols involve Alice sending coherent states taken from a typically small set, e.g. containing two or four states, according to some distribution, to Bob, who then applies a homo- or heterodyne measurement and discretises his outcome. Despite their simplicity, less is known about the security of such schemes.

The main challenge is that, unlike in the case of Gaussian modulation, the first and second moment of Alice and Bob's state are generally not sufficient to determine Eve's information, as one cannot invoke Gaussian extremality. Nevertheless, security against collective attacks has been shown in a number of scenarios. In \cite{leverrier2009unconditional}, security was proven for a limited class of transmission channels. For a protocol using Gaussian modulation for parameter estimation and discrete modulation for key generation, which requires decoy states, security against collective attacks and general security in the asymptotic limit has been shown in \cite{leverrier2011continuous}. The authors of \cite{ghorai2019asymptotic} apply an optimisation over possible covariance matrices of Alice and Bob's state as well as a reduction to the Gaussian optimality method to show security against collective attacks in the asymptotic limit. Higher key rates, which are secure against collective attacks in the asymptotic limit, are obtained by \cite{lin2019asymptotic}, which use an optimisation over all possible density matrices of Alice and Bob's state that are compatible with the observed statistics and without invoking the arguments of Gaussian optimality, but using a cutoff assumption that limits the number of photons in the state. This assumption was removed in \cite{Upadhyaya2021} for the asymptotic case, and in \cite{kanitschar2023finite} for the finite-size case; however both works only consider collective attacks. Security of discrete modulated coherent state protocols with homodyne or heterodyne detection against collective attacks was also shown by \cite{denys2021explicit,liu2021homodyne} in the asymptotic case, and by \cite{lupo2022quantum} in the finite-size case. In the limit of a high number of coherent states, asymptotic security against collective attacks was shown in \cite{kaur2019asymptotic}. In the setting of collective Gaussian attacks, the security of discrete modulated coherent state protocols with heterodyne detection, for any number of coherent states, has also been proven in the finite-size regime \cite{papanastasiou2019continuous}. Finally, finite-size security against general attacks has been shown for a protocol involving a discrete modulation using two coherent states ~\cite{matsuura2021finite,yamano2022finite,Matsuura2023}.

In this work, we consider a protocol involving a discrete modulation using four coherent states \cite{namiki2003security}, and heterodyne detection \cite{weedbrook2004quantum}, which is closely related to the protocol presented in \cite{lin2019asymptotic}. The main difference with respect to~\cite{lin2019asymptotic} is that all the information generated by the protocol, for key generation and parameter estimation, is discretised, in a similar way as was done in the entropic uncertainty relation approach of \cite{furrer2012continuous,furrer2014reverse}. This allows us to prove security against general attacks, as well as finite block sizes, using the entropy accumulation theorem (EAT) \cite{dupuis2016entropy,dupuis2019entropy}, which has previously been used to prove the security of device independent quantum key distribution against general attacks \cite{arnon2018practical,arnon2019simple,tan2020improved}.

The EAT is a powerful tool that allows one to lower bound the conditional smooth min-entropy, a quantity that can be used to quantify the amount of secret key obtainable from a (generally unstructured) classical-quantum (cq) state by means of privacy amplification, using hash functions \cite{renner2008security}. This is in fact the relevant situation in QKD protocols, since a cq-state is produced where Alice and Bob hold classical information, resulting in our case from Alice's preparation and Bob's measurements, whereas Eve's system remains quantum. The EAT requires the cq-state being the result of a sequence of maps, known as EAT channels, each of which provides classical outputs and side-information, while also passing on a quantum system to the next map. The lower bound on the conditional smooth min-entropy is in terms of a so-called `min-tradeoff function', mapping the observed statistic of classical outputs of the EAT channels to a real number which cannot exceed the single round conditional von Neumann entropies of any of the EAT channels.

A major challenge, when applying the EAT in security proofs for QKD, is that the EAT channels need to fulfill a Markov condition, ensuring that in each round, given all past-side information,  there are no new correlations between previous outcomes and the new side information. As information used for parameter estimation is obtained from measurements by Alice and Bob on systems which Eve could potentially have correlated in a way incompatible with the Markov condition, the EAT cannot be applied to the QKD protocol directly. Rather, a hypothetical EAT process is introduced which produces the same marginal states on the subsystems relevant to the security proof, and the smooth min-entropy of the QKD protocol is lower bounded using a combination of chain rules, as well as a min-tradeoff function corresponding to the EAT process \cite{arnon2018practical,arnon2019simple,tan2020improved}. 

As was recently pointed out by the authors of \cite{george2022finite}, another issue arises when applying the EAT in device dependent prepare-and-measurement protocols. Such protocols can be translated into entanglement based protocols, where Alice, instead of randomly sending states, prepares an entangled state, part of which is sent to Bob via an insecure channel, while the remaining part is kept in Alice's lab. Alice and Bob then perform measurements on their respective parts. The issue which arises is that the statistics obtained from the measurements is not sufficient to certify that the state between Alice and Bob is entangled, requiring additional constraints on Alice's marginal in the final key rate optimisation, which are incompatible with the EAT. We overcome this issue by adding an additional tomography performed by Alice in randomly chosen rounds, thus ensuring that Alice and Bob's measurement statistics are sufficient to certify entanglement between Alice and Bob. 

Having overcome these challenges, we are able to derive a min-tradeoff function using the numerical approach presented in \cite{winick2018reliable}, which requires a photon number cutoff assumption. It involves a linearisation of the objective function and the use of duality, finally reducing the problem to a semi-definite programming optimisation, which can be efficiently handled numerically. Our numerical analysis also suggests that the values of the key rate do not significantly vary once the cutoff becomes large enough. Using this approach, we are able to obtain asymptotic rates in the range of $0.1-10^{-4}$ bits per round for distances up to hundred kilometres. In the finite setting, and for realistic parameters, we get of the order of $10$ Gbits of secret key after $n\sim10^{11}$ rounds and distances of few tens of kilometres.

After most of the work that went into this result was completed, a generalised version of the EAT has been presented \cite{metger2022generalised,metger2022security}, which offers an alternative method of overcoming the challenges to prove the security of device-dependent prepare-and-measure protocols mentioned in the previous two paragraphs. In another recent result, the authors of \cite{kanitschar2023finite} have overcome the photon number cutoff assumption on Bob's state needed to compute the min-tradeoff function that defines the asymptotic rates by means of adding an additional energy test, as well as a dimension reduction technique presented in \cite{Upadhyaya2021}. Their proof also works in the finite setting, albeit only against collective attacks.

\section{Preliminaries}
\subsection{Basic notations}
In this section we introduce some definitions and concepts we use throughout the paper. For a Hilbert space $\mc{H}_A$, we denote by $\mc{D}(\mc{H}_A)$ the set of density operators, i.e. positive semidefinite operators with unit trace, $\rho_A$, acting on quantum system $A$. Sometimes it will be convenient to consider subnormalised states, i.e. states with $\Tr[\rho]\leq1$, in which case we use the notation $\mc{D}_{\leq}(\mc{H}_A)$.  The notation $\mc{H}_{AB}$ denotes a tensor product Hilbert space $\mc{H}_A\otimes\mc{H}_B$, and $\rho_{AB}$ the corresponding bipartite density operator. Classical random variables $X$, taking values $\{x\}$ according to the distribution $\{p_x\}$ can be expressed as density operators as $\rho_X=\sum_xp_x\pro{x}_X$. By $XY$ we denote the Cartesian product of random variables $X$ and $Y$. Further, we will be using the notation $A_1^n=A_1A_2...A_n$ and $X_1^n=X_1X_2...X_n$ for quantum and classical systems. We express cq states using the notation $\rho_{XA}=\sum_xp_x\pro{x}_X\otimes\rho_A^x$. For a cq state $\rho_{CQ} = \sum_{c} p(c) \pro{c} \otimes \rho_c$, an event $\Omega$ is defined as a subset of the elements $\{c\}$. The conditional state is then given by $\rho_{CQ}|_{\Omega} =\frac{1}{\Pr_{\rho}[\Omega]} \sum_{c\in\Omega} p(c) \pro{c} \otimes \rho_c$ where $\Pr_{\rho}[\Omega] := \sum_{c\in\Omega} p(c)$. When the state $\rho$ is clear from the context, we use $\Pr[\Omega]$ in place of $\Pr_{\rho}[\Omega]$. 

For two subnormalised states $\rho,\sigma\in\mc{D}_{\leq}(\mc{H}_A)$, we define the generalised fidelity
\begin{equation}
    F(\rho,\sigma)=\(\Tr\left|\sqrt{\rho}\sqrt{\sigma}\right|+\sqrt{(1-\Tr[\rho])(1-\Tr[\sigma])}\)^2,
\end{equation}
the generalised trace distance
\begin{equation}
    \Delta(\rho,\sigma)=\frac{1}{2}\|\rho-\sigma\|_1+\frac{1}{2}\left|\Tr[\rho-\sigma]\right|,
\end{equation}
as well as the purified distance
\begin{equation}
    P(\rho,\sigma)=\sqrt{1-F(\rho,\sigma)}.
\end{equation}
The generalised trace distance and the purified distance are metrics on $\mc{D}_{\leq}(\mc{H}_A)$. They are related by the Fuchs-van de Graaf inequality
\begin{align}
  \Delta(\rho,\sigma)\leq  P(\rho,\sigma)&\leq\sqrt{2\Delta(\rho,\sigma)-\Delta(\rho,\sigma)^2}\nonumber\\
  &\leq\sqrt{2\Delta(\rho,\sigma)}.\label{Fuchs}
\end{align}
In this work we make use of a number of entropic quantities. In addition to the well known von Neumann entropy, $H(A)_\rho=H(\rho_A)=-\Tr[\rho_A\log\rho_A]$, the conditional von Neumann entropy, $H(A|B)_{\rho}=H(AB)_{\rho}-H(B)_{\rho}$, as well as the Umegaki relative entropy,
\begin{equation}
    D(\rho||\sigma)=\frac{1}{\Tr[\rho]}\Tr\left[\rho(\log\rho-\log\sigma)\right] ,
\end{equation}
when $\supp(\rho)\subset\supp(\sigma)$ and $+\infty$ otherwise, for positive semidefinite $\rho$ and $\sigma$, we make use of min and max conditional entropies, defined for a subnormalised quantum state $\rho_{AB}\in\mc{D}_{\leq}(\mc{H}_{AB})$ by \cite{tomamichel2015quantum},
\begin{align}
    &H_{\min}(A|B)_\rho\nonumber\\
    &=\sup_{\sigma_B\in\mc{D}_{\leq}(\mc{H}_{B})}\sup\left\{\lambda\in\mathbb{R}:\rho_{AB}\leq\exp(-\lambda)\mathbb{1}_A\otimes\sigma_B\right\},\\
    &H_{\max}(A|B)_\rho=\max_{\sigma_B\in\mc{D}_{\leq}(\mc{H}_{B})}\log F\(\rho_{AB},\mathbb{1}_A\otimes\sigma_B\).
\end{align}
For $\epsilon\geq0$, we can then define the smooth min and max entropies as \cite{tomamichel2015quantum},
\begin{align}
    &H^\epsilon_{\min}(A|B)_\rho=\max_{\bar{\rho}\in\mc{B}^\epsilon(\rho_{AB})}H_{\min}(A|B)_{\bar{\rho}},\\
    &H^\epsilon_{\max}(A|B)_\rho=\min_{\bar{\rho}\in\mc{B}^\epsilon(\rho_{AB})}H_{\max}(A|B)_{\bar{\rho}},
\end{align}
where $\mc{B}^\epsilon(\rho_{A})$ is the $\epsilon$-ball around a state $\rho_A$ in terms of purified distance, i.e. the set of subnormalised states $\tau\in\mc{D}_{\leq}(\mc{H}_{A})$ such that $P(\tau,\rho)\leq\epsilon$. For parameter $a\in(1,2)$, let us further define the sandwiched R\'enyi divergence~\cite{muller2013quantum,wilde2014strong} for a quantum state $\rho$ and positive semidefinite $\sigma$ as
\begin{equation}
    D_a(\rho||\sigma)=    \frac{1}{a-1}\log\Tr\left[\(\sigma^{-\frac{a-1}{2a}}\rho\sigma^{-\frac{a-1}{2a}}\)^a\right] ,
\end{equation}
when $\supp(\rho)\subset\supp(\sigma)$ and $+\infty$ otherwise, and the conditional  R\'enyi entropy as
\begin{equation}
    H^\uparrow_a(A|B)_\rho=\inf_{\sigma_B\in\mc{D}_{\leq}(\mc{H}_{B})}D_a(\rho_{AB}||\mathbb{1}_A\otimes\sigma_B).
\end{equation}

\subsection{Security definition}

When two parties, Alice and Bob, wish to communicate in perfect secrecy in the presence of a quantum eavesdropper Eve, they need to perform a QKD protocol, typically consisting of $n$ rounds of quantum communication and local measurements, followed by classical post-processing steps involving parameter estimation, error correction and privacy amplification. An instance of a QKD protocol may be aborted if certain tests included in the protocol, such as parameter estimation, fail, or if a subprotocol, such as error correction aborts. If the protocol does not abort, the goal is to obtain a state close to a so-called perfect classical-classical-quantum (ccq) state of the form 
\begin{equation}
    \rho^\text{perfect ccq}_{K_AK_BE}=\frac{1}{d}\sum_{x=0}^{d-1}\pro{xx}_{K_AK_B}\otimes\rho_E ,
\end{equation}
where Alice and Bobs's systems are classical, whereas Eve's system may be quantum. Such a state corresponds to $\log d$ bits of an ideal classical key between Alice and Bob which is secret in that it is completely uncorrelated from Eve, even if Eve is allowed to possess a quantum system. And it is correct in the sense that Alice and Bob's systems are perfectly classically correlated.

A proof of security of a QKD protocol then involves two parts: Firstly, it has to be shown that it results in a state that is sound, i.e. close to a perfect ccq-state. Formally,  for $\epsilon^\mathrm{sou}>0$, a QKD protocol is said to be \emph{$\epsilon^\mathrm{sou}$-sound}, if it results in a state $\rho^\mathrm{QKD}_{K_AK_BE}$, such that if we condition on the event ${\Omega_\mathrm{NA}}$ of not aborting the protocol it holds
\begin{equation}
\Pr_{\rho^\mathrm{QKD}}[\Omega_{\mathrm{NA}}]\frac{1}{2}\left\|\rho^\mathrm{QKD}_{K_AK_BE}|_{\Omega_\mathrm{NA}}-\rho^\text{perfect ccq}_{K_AK_BE}\right\|_1\leq\epsilon^\mathrm{sou}.
\end{equation}
As we wish to treat the error correction protocol separately from the the remaining protocol, it is convenient to split the soundness property into a secrecy and correctness part. Namely, let $\epsilon^\mathrm{sec}>0$ and $\epsilon^\mathrm{cor}>0$. A QKD protocol is said to be \emph{$\epsilon^\mathrm{sec}$-secret} if 
\begin{equation}
\Pr_{\rho^\mathrm{QKD}}[\Omega_{\mathrm{NA}}]\frac{1}{2}\left\|\rho^\mathrm{QKD}_{K_AE}|_{\Omega_\mathrm{NA}}-\rho^\text{perfect ccq}_{K_AE}\right\|_1\leq\epsilon^\mathrm{sec}.
\end{equation}
The protocol is further said to be \emph{$\epsilon^\mathrm{cor}$-correct} if
\begin{equation}
    \Pr_{\rho^\mathrm{QKD}}[K_A\neq K_B\land\Omega_ \mathrm{NA}]\leq \epsilon^\mathrm{cor}.
\end{equation}
If the protocol is both $\epsilon^\mathrm{sec}$-secret and $\epsilon^\mathrm{cor}$-correct, it is $\epsilon^\mathrm{sec}+\epsilon^\mathrm{cor}$-sound. The second part of a security proof is to show completeness, meaning that there is an honest implementation, i.e. an implementation without presence of Eve, that does succeed, i.e. does not abort, with high probability. Formally, for $\epsilon^\mathrm{com}>0$, we say that a QKD protocol is \emph{$\epsilon^\mathrm{com}$-complete}, if
\begin{equation}
1-\Pr_\mathrm{hon}[\Omega_{\mathrm{NA}}]\leq\epsilon^\mathrm{com},
\end{equation}
where the subscript ``$\mathrm{hon}$" refers to the fact that we compute the probability with respect to the honest implementation specified by the protocol.

\section{The QKD protocol}\label{sec:Protocol}

The QKD protocol we consider is based on the four coherent-state protocol using heterodyne detection described in \cite{lin2019asymptotic}. However, we perform a discretisation of Bob's measurement outputs in both key and parameter rounds rather than just in key rounds. Our protocol also differs from the one presented in \cite{lin2019asymptotic} in that we do not include post-selection. 
In each round of the protocol, Alice prepares one of four coherent states $\ket{\varphi_x} \in \{\ket{\alpha},\ket{-\alpha},\ket{i \alpha}, \ket{-i\alpha}\}$ for some predetermined $\alpha\in\R$, with probability $\frac{1}{4}$. The state is then sent to Bob via a noisy channel that is potentially compromised by Eve. Bob then performs a heterodyne measurement. 

We will prove security using an equivalent entanglement-based QKD protocol. Such a protocol can be defined by the source replacement scheme \cite{bennett1992quantum,Grosshans2003Virtual,curty2004entanglement,ferenczi2012symmetries}. Namely, in each round  $i=1,...,n$, Alice prepares an independent copy of the pure state
\begin{equation}\label{eq:Initial}
\ket{\psi}_{AA'}=\frac{1}{2}\sum_{x=0}^3\ket{x}_A\ket{\varphi_x}_{A'}.
\end{equation}
Alice sends the $A'$ subsystem to Bob via a noisy quantum channel, keeping the $A$ subsystem. Alice and Bob then both perform measurements on their respective subsystems.  

Whereas this kind of introduction of an entanglement-based protocol is commonly used when proving security of prepare-and-measure protocols, we face an additional challenge when combining this approach with the EAT \cite{george2022finite}. Namely, unlike in a device independent setting, the statistics obtained from Alice and Bobs measurement, even in an honest implementation, is not sufficient to certify entanglement of eq. (\ref{eq:Initial}). In fact, as in the entanglement-based version Alice only implements one measurement in the computational basis, the statistics produced by the protocol can equally be explained by the separable state
\begin{equation}
\label{eq:Sepinitial}
\rho_{AA'}=\frac{1}{4}\sum_{x=0}^3\pro{x}_A\otimes\pro{\varphi_x}_{A'}.
\end{equation}
This is why to derive a positive secret-key rate, one includes a constraint on the marginal of Alice's state in the optimisation for the key rate~\cite{lin2019asymptotic} . Namely, the marginal is required to take the form 
\begin{equation}\label{eq:InitialA}
\rho_{A}=\frac{1}{4}\sum_{x,y=0}^3\<\varphi_y|\varphi_x\>\ket{x}\bra{y}_A,
\end{equation}
which is not satisfied by the separable state~\eqref{eq:Sepinitial}. The challenge then is to express such a constraint in terms of a distribution obtained by statistical analysis, which is required when applying the EAT. 

We overcome this challenge by considering a hypothetical version of our protocol, where, in randomly chosen rounds, Alice performs tomographic measurements of her marginal on $A$ and, in the end, verifies whether the obtained statistics are compatible with her marginal being equal to eq. (\ref{eq:InitialA}). If this is not the case, the protocol gets aborted. As the data obtained in the tomography rounds is not used for key generation, the key rate obtained in this hypothetical protocol is never larger than the key rate obtained in the physically implemented protocol, where Alice performs no tomography. Also, as we are in a device dependent setting, where we can assume Alice's state preparation to be perfect, the only scenario under which the hypothetical protocol aborts after the tomography test is due to imperfect tomography, the probability of which becomes negligible for large enough $n$. In the following, we use the term `hypothetical QKD protocol' when we consider the entanglement-based protocol including tomography and `physical QKD protocol' when referring to the entanglement-based protocol which does not include tomography. By the source-replacement scheme the latter is equivalent to the prepare-and-measure protocol that is actually performed in the laboratory by Alice and Bob.

When the state eq. (\ref{eq:Initial}) is sent from Alice to Bob, we assume that Eve can attack the channel used to send the $A'$ subsystem coherently. This is equivalent to a scenario where Alice initially prepares all $n$ independent and identically distributed (iid) copies of the state (\ref{eq:Initial}), which are then acted upon a channel $\mc{N}_{{A'}_1^n\to B_1^n}$. Let $\mc{U}^\mc{N}_{A'^n\to B^nE}$ be an isometric extension of the channel and let us define
\begin{equation}\label{eq:initialstateABE}
\ket{\Psi}_{A_1^nB_1^nE}=\id_{A^n}\otimes\,\mc{U}^\mc{N}_{{A'}_1^n\to B_1^nE}\ket{\psi}^{\otimes n}_{A_1^n{A'}_1^n}.
\end{equation}
It has to be assumed that the $E$ subsystem goes to Eve. Alice and Bob are left with the mixed state
\begin{align}\label{eq:initialstate}
\rho_{A_1^n B_1^n}=\Tr_E [\Psi_{A_1^nB_1^nE}].
\end{align}

\subsection{The hypothetical QKD protocol}

We now describe a round of the hypothetical QKD protocol in detail. Let $0\leq p^\text{key}\leq1$, $0\leq p^\text{PE}\leq1$ and $0\leq p^\text{tom}\leq1$, where $p^\text{key}+p^\text{PE}+p^\text{tom}=1$, be the respective probabilities for a given round being used for key generation, parameter estimation and tomography of Alice's marginal.  For each round $i=1,...,n$, Alice and Bob perform the following steps:
\linebreak

(1) {\it Alice's Measurement:}
Alice uses a random number generator to create a random variable $R_i$, taking values $R_i=0,1,2$ with respective probabilities $p^\text{key}$, $p^\text{PE}$ and $p^\text{tom}$. If $R_i=0$, the round is used for key generation. For $R_i=1$, the round is employed for parameter estimation. In both cases Alice  performs a projective measurement $\{\pro{x}\}_{x=0}^3$ on subsystem $A_i$. If $R_i=2$, Alice performs a tomography, using an informationally complete (IC) measurement defined by a Positive-Operator-Valued-Measure (POVM) $\{\Gamma_{x'}\}_{x'=0}^{15}$ on her subsystem. The outcome of Alice's measurement is described by a random variable $X_i$, taking values $x_i$. We define, for the sake of convenience, the random variables 
\begin{align}
\hat{X}_i&=\begin{cases}
x_i\text{ if }R_i=0,\\
\perp\text{ else}.
\end{cases}\\
\tilde{X}_i&=\begin{cases}
x_i\text{ if }R_i=1,\\
\perp\text{ else}.
\end{cases}\\
X'_i&=\begin{cases}
x_i\text{ if }R_i=2,\\
\perp\text{ else}.
\end{cases}
\end{align}
The random variable $R_i$ is then sent to Bob via an authenticated channel.
\newline

(2) {\it Bob's Measurement:}
Bob performs a heterodyne measurement on subsystem $B_i$. From the outcome, Bob obtains a continuous random variable $Y_i$, taking values 
$y_i\in\mathbb{C}$. Again, it will be convenient to define
\begin{align}
\hat{Y}_i&=\begin{cases}
y_i\text{ if }R_i=0,\\
\perp\text{ else}.
\end{cases}\\
\tilde{Y}_i&=\begin{cases}
y_i\text{ if }R_i=1,\\
\perp\text{ else}.
\end{cases}
\end{align}

(3) {\it Discretisation:} Bob discretises his heterodyne outcomes. For key rounds, let $\hat{y}_i=|\hat{y}_i|e^{i\hat{\theta}_i}$ for $\hat{\theta}_i\in[-\frac{\pi}{4},\frac{7\pi}{4})$. Bob then creates a random variable
\begin{equation}\label{eq:disckey}
\hat{Z}_i=\begin{cases}
0\text{ if }\hat{\theta}_i\in[-\frac{\pi}{4},\frac{\pi}{4})\\ 
1\text{ if }\hat{\theta}_i\in[\frac{\pi}{4},\frac{3\pi}{4})\\
2\text{ if }\hat{\theta}_i\in[\frac{3\pi}{4},\frac{5\pi}{4})\\
3\text{ if } \hat{\theta}_i\in[\frac{5\pi}{4},\frac{7\pi}{4})\\
\perp\text{else},
\end{cases}
\end{equation}
where $\hat{Z}_i=\perp$ is taken for non-key rounds. For parameter estimation rounds, Bob defines a discretisation given by an amplitude $\Delta$ and modules of length $\delta$, such that $\Delta/\delta \in \mathbb{N}$. Let $j \in \{0,1,2,3 \}$ and $k\in\{0, \dots, \frac{\Delta}{\delta}-1\}$ and let $\tilde{y}_i=|\tilde{y}_i|e^{i\tilde{\theta}_i}$. Bob then creates a random variable $\tilde{Z}_i$ according to
\begin{equation}\label{eq:discPE}
\tilde{Z}_i=\begin{cases}
j + 4k & \text{ if }  \substack{\tilde{\theta}_i\in[\frac{\pi}{4}(2j-1),\frac{\pi}{4}(2j+1)) \\ |\tilde{y}_i| \in [\delta k, \delta (k+1)),}\\
j + 4 \frac{\Delta}{\delta} & \text{ if } \substack{\tilde{\theta}_i\in[\frac{\pi}{4}(2j-1),\frac{\pi}{4}(2j+1)) \\ |\tilde{y}_i|\in [\Delta,\infty),} \\
\perp & \text{ else}.
\end{cases}
\end{equation}

\begin{figure}
    \centering
    \includegraphics[width=0.98\linewidth]{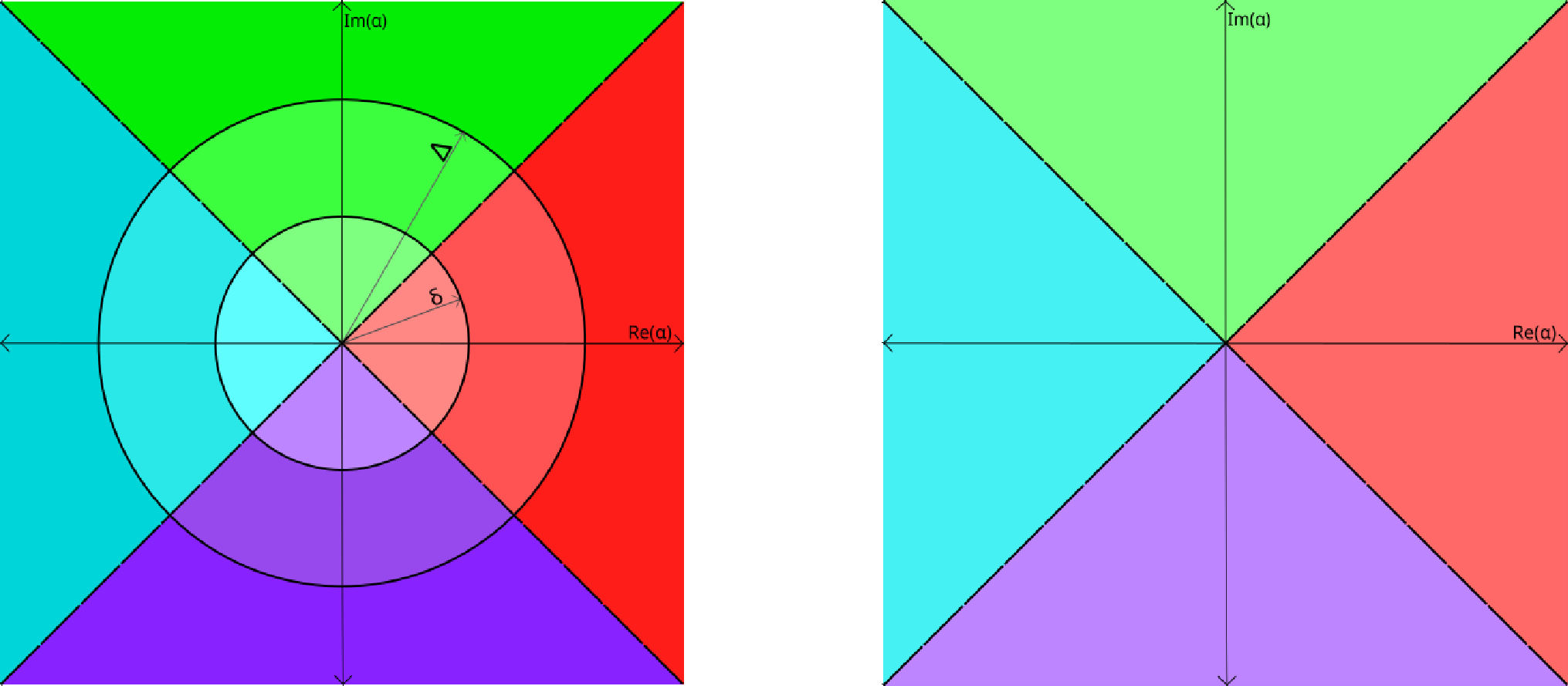}
    \caption{Discretisations of phase space by Bob for parameter estimation rounds (left) and key generation rounds (right). In this figure, the modulation for parameter estimation in phases and amplitudes is given by $\Delta/\delta = 2$, with the outmost modules extending to infinity.}
    \label{fig:PhaseSpaceModules}
\end{figure}

Summarising steps (1) - (3), round $i$ of the protocol has taken as inputs quantum systems $A_iB_i$ of the initial state (\ref{eq:initialstate}), and created discrete classical random variables $\hat{X}_i$ and $\hat{Z}_i$ for key generation,  $\tilde{X}_i$ and $\tilde{Z}_i$ to be used for parameter estimation, as well as $X'_i$ to be used for tomography of Alice's marginal state. Let us define $O_i:=\tilde{X}_iX'_i\hat{Z}_i\tilde{Z}_i$ as the `output' and $S_i:=R_i$ as the `side information'. Let us further define $C_i=\tilde{X}_i\tilde{Z}_iX'_i$ as all the information used in statistical analysis. The reason we define $O_i$, $S_i$ and $C_i$ in this way is that we will later use these random variables when applying the EAT. The EAT requires the statistical analysis variable $C_i$ to be obtainable from a simple read-out of `output' and `side-information' variables $O_i$ and $S_i$. On the other hand, we cannot include $\tilde{X}_i$, $\tilde{Z}_i$ or $X'_i$ into $S_i$ because of the Markov condition, eq. (\ref{eq:Markov}). We therefore have to include them into $O_i$ despite the fact that $\tilde{X}_i$ has to be communicated classically, and treat $\tilde{X}_i$ as additional side-information when applying Proposition \ref{renner}.

Any round $i$ of the protocol can then be described by a channel
\begin{equation}\label{def:M123}
\mc{M}^\mathrm{QKD}:A_iB_i \to \hat{X}_iO_iS_iC_i.
\end{equation}
After $n$ rounds, the relevant systems of Alice, Bob and Eve are in the state
\begin{equation}\label{eq:15}
\sigma^\mathrm{QKD}_{\hat{X}_1^nO_1^nC_1^nS_1^nE}=\id_E\otimes{\mc{M}^\mathrm{QKD}}^{\otimes n}\({\Psi}_{A_1^nB_1^nE}\).
\end{equation}

The next step is to perform parameter estimation. To that purpose, Alice sends $\tilde{X}_1^n$ to Bob, who then performs the parameter estimation protocol deciding whether the protocol gets aborted or not. Further, Alice uses ${X'}_1^n$, obtained from her tomographic measurement, to reconstruct her marginal state. If the reconstructed state is not equal to the expected one up to a certain margin of confidence, the protocol is aborted. Alice informs Bob of her decision, and in the latter case the protocol is aborted.

In order to formalise the decision to abort, we need to introduce some notation. Let us denote by $\mc{C}$ the alphabet of all possible values of $c_i=(\tilde{x}_i,\tilde{z}_i,x'_i)$ that can occur in the protocol. Such values are $c_i=(\perp,\perp,\perp)$ in key rounds, $(x,z,\perp)$ with $x\in\{0,...,3\}$ and $z\in\{0,...,m-1\}$, (where $m=4 \Delta/\delta + 4$ is the total number of modules) in parameter estimation rounds, as well as $(\perp,\perp,x')$ with $x'\in\{0,...,15\}$ in tomography rounds. It will also be convenient to define by $\tilde{\mc{C}}$ the alphabet of all possible values $c_i$ can take in parameter estimation and tomography rounds, only. For a given string $c_1^n \in \mc{C}^n$, we denote by $\freq_{c_1^n}\in\mc{P}_\mc{C}$ the probability distribution corresponding to the frequency of symbols $c\in\mc{C}$ in $c_1^n$, defined by
$\freq_{c_1^n}(c)={|\{i:c_i=c\}|}/{n}$.

In order to decide whether or not to abort, Alice and Bob need to benchmark their obtained statistics, given by a frequency distribution $\freq_{c_1^n}$, against a distribution $p_0\in\mc{P}_{\mc{C}}$, which can be obtained in an honest implementation of the protocol. Let $p_0^\text{sim}$ be the distribution of parameter estimation random variables $(\tilde{X}, \tilde{Z})$ in the honest setting with no attack and $p_0^{\text{tom}}$ be the distribution of the tomography random variable $X'$. Now, we define
\begin{align}
    &p_0(x,z,\perp)=
    p^\text{PE}p_0^\text{sim}(x,z),\label{eq:p0PE}\\
     &p_0(\perp,\perp,x')
     =p^\text{tom}p_0^\text{tom}(x')\label{eq:p0tom},\\
     &p_0(\perp,\perp,\perp)=1-\sum_{xz}p_0(x,z,\perp)-\sum_{x'}p_0(\perp,\perp,x')
\end{align}
for $x\in\{0,...,3\}$, $z\in\{0,...,m-1\}$, and $x'\in\{0,...,15\}$. 
We will provide an explicit form of the $p_0^\text{sim}$ and $p_0^\text{tom}$ we use in Section \ref{subs:NumImplementation}. 

In order to compare the two distributions $\freq_{c_1^n}$ and $p_0\in\mc{P}_{\mc{C}}$, we need to introduce  figures of merit, which quantify the suitability of the distributions for key generation. For now, let us only assume that these figures of merit are given by  affine functions $f^\text{PE}:\mc{P}_{\mc{C}}\to\mathbb{R}$ and $f^\text{tom}:\mc{P}_{\mc{C}}\to\mathbb{R}$ of the form
\begin{align}
    &f^\mathrm{PE}(p)=\sum_{x=0}^3\sum_{z=0}^{m-1}h_{x,z,\perp}p(x,z,\perp),\label{eq:fPE}\\
     &f^\mathrm{tom}(p)=\sum_{x'=0}^{15}h_{\perp,\perp,x'}p(\perp,\perp,x'),\label{eq:ftom}
    \end{align}
for some coefficients $h_{x,z,\perp},h_{\perp,\perp,x'}\in\mathbb{R}$. We will provide an explicit form of the functions later, together with the methodology that compares $\freq_{c_1^n}$ and $p_0$. For the moment, let us note that a reduced number of binnings (i.e., employing as few modules as possible for parameter estimation) decreases the overall differences of distributions $\freq_{c_1^n}$ and $p_0$ in the case of $f^\mathrm{PE}$. This is eventually reflected as an increase in $p^\mathrm{key}$, since fewer rounds must be spent on bounding the differences between said distributions. We now define the respective sets of probabilities for which we do not abort after parameter estimation or tomography as
\begin{align}
&\mc{P}_{\Omega_\mathrm{PE}}:=\left\{p\in\mc{P}_\mc{C}:f^\mathrm{PE}(p)\geq f^\mathrm{PE}(p_0)-\delta_\mathrm{PE}^\mathrm{tol}\right\},\label{def:OmegaEA1}\\
&\mc{P}_{\Omega_\mathrm{tom}}:=\left\{p\in\mc{P}_\mc{C}:f^\mathrm{tom}(p)\geq f^\mathrm{tom}(p_0)-\delta_\mathrm{tom}^\mathrm{tol}\right\},
\end{align}
for some $\delta_\mathrm{PE}^\mathrm{tol},\delta_\mathrm{tom}^\mathrm{tol}>0$. Let us also define $\delta^\mathrm{tol}=\delta_\mathrm{PE}^\mathrm{tol}+\delta_\mathrm{tom}^\mathrm{tol}$, as well as the events of passing the parameter estimation and the tomography test as
\begin{align}
    &\Omega_\mathrm{PE}:=\left\{c_1^n\in\mc{C}^n:\freq_{c_1^n}\in\mc{P}_{\Omega_\mathrm{PE}}\right\},\label{eq:OmegaPE}\\
    &\Omega_\mathrm{tom}:=\left\{c_1^n\in\mc{C}^n:\freq_{c_1^n}\in\mc{P}_{\Omega_\mathrm{tom}}\right\},\label{eq:Omegatom}\\
    &\Omega_\mathrm{EA}=\Omega_\mathrm{PE}\cap\Omega_\mathrm{tom}.\label{def:OmegaEA2}
\end{align}

Assuming that Alice and Bob do not abort after parameter estimation or tomography, they perform an error correction protocol using reverse reconciliation. The information exchanged between Alice and Bob in this step is denoted $L$, and at the end of this step Alice computes a string $\bar{X}_1^n$. In order to check that the error correction was successful, Bob chooses a random hash function $H$ and sends to Alice a description of $H$ as well as the value $H' = H(\hat{Z}_1^n)$. Whenever $H(\hat{Z}_1^n) \neq H(\bar{X}_1^n)$, the protocol is aborted. Let us denote by $H$ and $H'$ the register containing the description and value of the hash function, respectively. Formally, we define the event of passing the error correction step as

\begin{equation}\label{eq:OmegaEC}
    \Omega_{\mathrm{EC}} = \big[H(\hat{Z}_1^n) = H(\bar{X}_1^n)\big].
\end{equation}

We assume that there is a small probability, upper bounded by $\epsilon_\mathrm{EC}>0$,  of the error correction being passed by mistake. For any $\hat{z}_1^n \neq \bar{x}_1^n$, $\Pr[ H(\hat{z}_1^n) = H(\bar{x}_1^n) ] \leq \epsilon_\mathrm{EC}$, where Pr here is over the choice of $H$. Further we assume that the probability of not passing the error correction in an honest implementation is upper bounded by $\Pr_\mathrm{hon}\big[H(\hat{Z}_1^n) \neq H(\bar{X}_1^n)\big]\leq \epsilon^c_\mathrm{EC}$, for some $\epsilon^c_\mathrm{EC}>0$.

\begin{protocol}{Hypothetical QKD protocol}

\begin{enumerate}

\item {Alice prepares state $\ket{\Psi}^{\otimes n}_{AA'}$ given by eq. (\ref{eq:Initial}) and sends subsystems ${A'}_1^n$ to Bob via a noisy channel, and Bob receives subsystems $B_1^n$. The total state is then given by eq. (\ref{eq:initialstateABE}).}

  \item {For round $i=1,...,n$ the following steps are performed:}
  \begin{enumerate}
    \item{Alice chooses $R_i\in\{0,1,2\}$ according to $(p^\mathrm{key},p^\mathrm{PE},p^\mathrm{tom})$ and sends $R_i$ to Bob.}

    \item{If $R_i=0$, Alice measures $A_i$ in a computational basis and stores output in $\hat{X}_i$. Bob measures $B_i$ using a heterodyne measurement, discretises according to eq. (\ref{eq:disckey}), and stores the result in $\hat{Z}_i$.}

 \item{If $R_i=1$, Alice measures $A_i$ in computational basis and stores output in $\tilde{X}_i$. Bob measures $B_i$ using heterodyne measurements, discretises according to eq. (\ref{eq:discPE}), and stores the result in $\tilde{Z}_i$.}

 \item{If $R_i=2$, Alice measures $A_i$ using an informationally complete measurement and stores output in $X'_i$.}
 
    \end{enumerate}
    
  \item {Alice and Bob use $\tilde{X}_1^n\tilde{Z}_1^n$ for parameter estimation. Alice uses ${X'}_1^n$ for the tomograpahy test. If either fails the protocol is aborted.}
  
  \item {If the protocol has not been aborted, error correction is performed using reverse reconciliation. If error correction fails, the protocol is aborted. Otherwise privacy amplification is performed resulting in the final key.}

\end{enumerate}
\end{protocol}\label{prot:Protocol1}

Now, we define the event of not aborting the protocol after either parameter estimation, tomography or error correction as \begin{equation}\label{eq:OmegaNonAbort}
\Omega_{\mathrm{NA}} = \Omega_{\mathrm{EA}} \cap \Omega_{\mathrm{EC}}.
\end{equation}  
We note that $\Omega_{\mathrm{EA}}$ only depends on the $C_1^n$ registers, whereas $\Omega_{\mathrm{EC}}$ depends on the $HH'\bar{X}_1^n$ registers. The description of the protocol together with an attack of Eve leads to the final state

\begin{align*}
&\sigma^\mathrm{QKD}_{\bar{X}_1^n O_1^n S_1^n C_1^nHH'LE}\\
&=\Pr[\Omega_{\mathrm{NA}}]\sigma^\mathrm{QKD}_{\bar{X}_1^n O_1^n S_1^n C_1^nHH'LE}|_{ \Omega_{\mathrm{NA}}}\\
&\quad +(1-\Pr[\Omega_{\mathrm{NA}}])\sigma^\mathrm{QKD}_{\bar{X}_1^n O_1^n S_1^n C_1^nHH'LE}|_{ \neg\Omega_{\mathrm{NA}}}.
\end{align*}
At this point, if Alice and Bob did not abort, they proceed with a privacy amplification protocol (e.g. via two-universal hash functions) to distill the final, secret key. 

\subsection{The physical QKD protocol}
Finally, let us describe the physical QKD protocol, which is the equivalent entanglement-based version of the prepare-and-measure protocol that is actually performed in a realistic implementation. The physical QKD protocol is essentially equal to the hypothetical protocol except for the following two differences: Firstly, in step (1) of the protocol, whenever $R_i=2$ Alice does not perform tomography---the round is simply discarded. Hence, the random variable $X'_i$ will be either $\perp$ or undefined. Secondly, as a consequence of not performing the tomography, abortion or non-abortion at the end of the protocol will only be determined by parameter estimation and error correction. I.e., the event $\Omega_\mathrm{NA}$ will be replaced by $\Omega^\mathrm{phys}_\mathrm{NA}=\Omega_\mathrm{PE}\cap\Omega_\mathrm{EC}$. 

In principle, a more efficient physical protocol can be obtained if no round is discarded. That is, if we set $p^\text{tom}=0$. In our case, however, a non-zero value of $p^\text{tom}$ is chosen because the comparison of the two protocols, hypothetical and physical, results in a much simpler analysis when using the same values for the probabilities $p^\text{key}$, $p^\text{PE}$ and $p^\text{tom}$ in both protocols. It is nevertheless worth noting that the value taken for $p^\text{tom}$ below is very small, so the possible impact on the key rate does not represent a significant loss.

All in all, the physical protocol is composed by the following steps

\begin{protocol}{Physical QKD protocol}

\begin{enumerate}

\item {Alice prepares state $\ket{\Psi}^{\otimes n}_{AA'}$ given by eq. (\ref{eq:Initial}) and sends subsystems ${A'}_i^n$ to Bob via a noisy channel, and Bob receives subsystems $B_1^n$. The total state is then given by eq. (\ref{eq:initialstateABE}).}

  \item {For round $i=1,...,n$ the following steps are performed:}
  \begin{enumerate}
    \item{Alice chooses $R_i\in\{0,1,2\}$ according to $(p^\mathrm{key},p^\mathrm{PE},p^\mathrm{tom})$ and sends $R_i$ to Bob.}

    \item{If $R_i=0$, Alice measures $A_i$ in a computational basis and stores output in $\hat{X}_i$. Bob measures $B_i$ using a heterodyne measurement, discretises according to eq. (\ref{eq:disckey}), and stores the result in $\hat{Z}_i$.}

\item{If $R_i=1$, Alice measures $A_i$ in a computational basis and stores output in $\tilde{X}_i$. Bob measures $B_i$ using a heterodyne measurement, discretises according to eq. (\ref{eq:discPE}), and stores the result in $\tilde{Z}_i$.}

 \item{If $R_i=2$, the round is not used.}
 
    \end{enumerate}
    
  \item {Alice and Bob use $\tilde{X}_1^n\tilde{Z}_1^n$ for parameter estimation. If it fails, the protocol is aborted.}
  
  \item {If the protocol has not been aborted, error correction is performed using reverse reconciliation. If error correction fails, the protocol is aborted. Otherwise privacy amplification is performed resulting in the final key.}

\end{enumerate}

\end{protocol}

\section{Security of the QKD protocol}

In this section we show the security against coherent attacks of the physical QKD protocol. The proof consists of two parts: Firstly, in Subsection \ref{sec:sound} we show the soundness of the  protocol and provide a lower bound on the key rate. We first show the soundness of the hypothetical protocol, which we then show implies the soundness of the physical protocol. The soundness proof of the hypothetical protocol is based on the EAT and depends on the choice of a min-tradeoff function of a particular form. Secondly, in Subsection \ref{sec:compl}, we show the completeness of the physical QKD protocol, i.e. that there is a nonzero probability of it not being aborted. Finally, in Subsection \ref{sec:MinTrade}, we show how a suitable min-tradeoff function can be derived from
 the numerical approach presented by \cite{lin2019asymptotic,winick2018reliable}.

 \subsection{Soundness}\label{sec:sound}

In this section we provide a lower bound on the achievable key rate $r^\mathrm{phys}=\ell/n$, where $\ell$ is the length of the key and $n$ the number of rounds, conditioned on the event $ \Omega^\mathrm{phys}_{\mathrm{NA}}$ of not aborting the physical QKD protocol. Such a lower bound can be obtained from  Proposition~\ref{renner} below, which is based on the leftover-hash Lemma \cite{tomamichel2015quantum}. To derive this proposition, we make the following

\textbf{Bounded energy assumption:} the attack by the eavesdropper involves states of finite energy.

Under this assumption, Eve's attack can be arbitrarily approximated by an attack using finite-dimensional systems. Or, in other words, we can assume Eve's attack to be defined in a Hilbert space of arbitrary finite dimension. This allow us to use the original version of the leftover hash lemma. We leave for future work how to generalise the security proof to infinite Hilbert space dimensions, for which a version of the leftover-hash Lemma has been derived in \cite{furrer2012security,berta2016smooth}.

\begin{proposition}\label{renner}\cite{renner2008security,tomamichel2015quantum} 
Let $\epsilon^\mathrm{phys},\epsilon_\mathrm{EC}\geq0$. Let further $\mathrm{leak}_\mathrm{EC}$ be the amount of information lost to Eve during error correction. Then Alice and Bob are able to extract a key of length,
\begin{align}
\ell\leq & H^{\epsilon^\mathrm{phys}}_{\min}(\hat{Z}_1^n|R_1^n\tilde{X}_1^nE)_{\sigma^\mathrm{phys,QKD}|_{ \Omega^\mathrm{phys}_{\mathrm{NA}}}}-\mathrm{leak}_\mathrm{EC}\nonumber\\
&-2\log\frac{1}{\epsilon^\mathrm{phys}},
\end{align}
which is $3\epsilon^\mathrm{phys}+\epsilon_\mathrm{EC}$-sound, in the sense that \begin{align*}
&\Pr_{\sigma^\mathrm{phys,QKD}}[\Omega^\mathrm{phys}_\mathrm{NA}]\frac{1}{2}\|\sigma^\mathrm{phys,QKD}|_{ \Omega^\mathrm{phys}_{\mathrm{NA}}}-\sigma^\mathrm{perfect \ ccq}\|_1\\
&\leq 3\epsilon^\mathrm{phys}+\epsilon_\mathrm{EC}.
\end{align*}

\end{proposition}

In order to apply Proposition \ref{renner}, we need to lower bound the smooth min-entropy using the EAT. However, as noted in the introduction, we are unable to apply the EAT directly to our physical QKD protocol due to the need to characterise Alice's marginal system in a prepare-and-measure scenario. To that purpose we first consider the hypothetical QKD protocol that includes additional tomography measurements. However, due to issues with the Markov condition, we will not be able to directly apply the EAT to our hypothetical protocol either. Instead, we will make use of various chain rules for smooth entropies in order to relate the output of our hypothetical QKD protocol to that of a series of $n$ EAT channels, which we call the `EAT process', and then apply the EAT to the EAT process, while dealing with the remaining terms separately.

 We begin by considering the hypothetical QKD protocol. Let $n\in\mathbb{N}$ and $\epsilon>0$. Conditioned on not aborting, the hypothetical QKD protocol results in the state $\sigma^\mathrm{QKD}|_{{ \Omega_{\mathrm{NA}}}}$. By application of chain rules for smooth entropies (eq. (6.63) and eq. (6.56) in \cite{tomamichel2015quantum}), it holds
\begin{align}
& H^\epsilon_{\min}(\hat{Z}_1^n|R_1^n\tilde{X}_1^nE)_{\sigma^\mathrm{QKD}|_{{ \Omega_{\mathrm{NA}}}}}\nonumber\\
&\geq H^{\epsilon/4}_{\min}(\hat{Z}_1^n\tilde{X}_1^n|R_1^nE)_{\sigma^\mathrm{QKD}|_{{ \Omega_{\mathrm{NA}}}}}\nonumber\\
&\quad -H^{\epsilon/4}_{\max}(\tilde{X}_1^n|R_1^nE)_{\sigma^\mathrm{QKD}|_{{ \Omega_{\mathrm{NA}}}}}-2\Gamma(\epsilon/4),\label{eq:20}
\end{align} 
where $\Gamma(x):=-\log\(1-\sqrt{1-x^2}\)$. By another application of a chain rule (eq. (6.57) in \cite{tomamichel2015quantum}), we obtain
\begin{align}
&H^{\epsilon/4}_{\min}(\hat{Z}_1^n\tilde{X}_1^n|R_1^nE)_{\sigma^\mathrm{QKD}|_{ \Omega_{\mathrm{NA}}}}\nonumber\\
&\geq H^{\epsilon/16}_{\min}(\tilde{X}_1^n {X'}_1^n \hat{Z}_1^n \tilde{Z}_1^n|R_1^nE)_{\sigma^\mathrm{QKD}|_{ \Omega_{\mathrm{NA}}}}-3\Gamma(\epsilon/16)\nonumber\\
&\quad -H^{\epsilon/16}_{\max}({X'}_1^n\tilde{Z}_1^n|\hat{Z}_1^n\tilde{X}_1^nR_1^nE)_{\sigma^\mathrm{QKD}|_{ \Omega_{\mathrm{NA}}}}\label{eq:22}.
\end{align}

We can now apply the same argument as used in \cite{dupuis2016entropy} to upper bound the max entropy terms in eqs. (\ref{eq:20}) and (\ref{eq:22}). We begin by upper bounding the term $H^{\epsilon/4}_{\max}(\tilde{X}_1^n|R_1^nE)_{\sigma^\mathrm{QKD}|_{{ \Omega_{\mathrm{NA}}}}}$ in eq. (\ref{eq:20}). We note that by the strong subadditivity of the smooth max entropy \cite{tomamichel2015quantum}, it holds
\begin{equation}
H^{\epsilon/4}_{\max}(\tilde{X}_1^n|R_1^nE)_{\sigma^\mathrm{QKD}|_{{ \Omega_{\mathrm{NA}}}}}\leq H^{\epsilon/4}_{\max}(\tilde{X}_1^n|R_1^n)_{\sigma^\mathrm{QKD}|_{{ \Omega_{\mathrm{NA}}}}},
\end{equation}
where the r.h.s. only involves classical registers. We further note that $\tilde{X}_i=\perp$, unless $R_i=1$, which happens with probability $p^\text{PE}$, in which case $\tilde{X}_i$ takes a value in $\{0,...,3\}$. Introducing a binary random variable $\tilde{R}_i$ that takes value $1$ when Alice's random variable is $R_i=1$, and takes value $0$ when $R_i=0$ or $R_i=2$, we can apply the data processing inequality and Lemma \ref{lem:Omar}  in Appendix \ref{App:B}, showing that
\begin{align}
 &H^{\epsilon/4}_{\max}(\tilde{X}_1^n|R_1^n)_{\sigma^\mathrm{QKD}|_{{ \Omega_{\mathrm{NA}}}}}\nonumber\\
 &\leq  H^{\epsilon/4}_{\max}(\tilde{X}_1^n|\tilde{R}_1^n)_{\sigma^\mathrm{QKD}|_{{ \Omega_{\mathrm{NA}}}}}\nonumber\\
 &\leq np^\text{PE}\log{5} + \sqrt{\frac{n}{2}\ln{\frac{32}{\epsilon^2 \Pr_{\sigma^\mathrm{QKD}}[\Omega_{\mathrm{NA}}]}}}\log{5}.
\end{align}
In a similar way we can provide an upper bound on the term $H^{\epsilon/16}_{\max}({X'}_1^n\tilde{Z}_1^n|\hat{Z}_1^n\tilde{X}_1^n R_1^nE)_{\sigma^\mathrm{QKD}|_{ \Omega_{\mathrm{NA}}}}$ in eq. (\ref{eq:22}). Again, it holds by strong subadditivity,
\begin{align}
&H^{\epsilon/16}_{\max}({X'}_1^n\tilde{Z}_1^n|\hat{Z}_1^n\tilde{X}_1^nR_1^nE)_{\sigma^\mathrm{QKD}|_{ \Omega_{\mathrm{NA}}}}\nonumber\\
&\leq H^{\epsilon/16}_{\max}({X'}_1^n\tilde{Z}_1^n|R_1^n)_{\sigma^\mathrm{QKD}|_{ \Omega_{\mathrm{NA}}}},
\end{align}
where the r.h.s. is classical. Further, it holds that ${X'}_i\tilde{Z}_i=\perp\perp$, unless $R_i=1$ or $R_i=2$, which happens with probability $p^\text{PE}+p^\text{tom}=1-p^\text{key}$. In this case ${X'}_i\tilde{Z}_i$ takes a value in $\{0,...,15,\perp\}\times\{0,...,m-1,\perp\}$. Let us again introduce a binary random variable $\bar{R}_i$, taking value $1$ when $R_i=1$ or $R_i=2$ and value $0$ when $R_i=0$. We can now apply Lemma \ref{lem:Omar}, identifying the pair ${X'}_i\tilde{Z}_i$ with $X_i$ and the value $\perp\perp$ with $\perp$, and obtain 
\begin{align}
 &H^{\epsilon/16}_{\max}({X'}_1^n\tilde{Z}_1^n|R_1^n)_{\sigma^\mathrm{QKD}|_{ \Omega_{\mathrm{NA}}}}\nonumber\\
 &\leq  H^{\epsilon/16}_{\max}({X'}_1^n\tilde{Z}_1^n|\bar{R}_1^n)_{\sigma^\mathrm{QKD}|_{ \Omega_{\mathrm{NA}}}}\\
 &\leq n(1-p^\text{key})\log{\(17(m+1)\)}\nonumber\\
 &\quad +\sqrt{\frac{n}{2}\ln{\frac{512}{\epsilon^2 \Pr_{\sigma^\mathrm{QKD}}[\Omega_{\mathrm{NA}}]}}}\log{(17(m+1))}.
\end{align}

What remains to be done is to lower bound the term $H^{\epsilon/16}_{\min}(O_1^n|R_1^nE)_{\sigma^\mathrm{QKD}|_{ \Omega_{\mathrm{NA}}}}$ in eq. (\ref{eq:22}) using the EAT.

\subsubsection{Reduction to Collective Attacks via Entropy Accumulation}

In order to apply the EAT, we wish to condition on an event that is only defined on the statistical information $C_1^n$, where we recall that $C_i = \tilde{X}_i \tilde{Z}_i X_i'$ is given by classical information extracted from the outputs $O_i = \tilde{X}_i X'_i \hat{Z}_i \tilde{Z}_i$ and the side information $S_i = R_i$. For such conditioning, note that we can write $\sigma^\mathrm{QKD}|_{ \Omega_{\mathrm{NA}}}$ as $(\sigma^\mathrm{QKD}|_{\Omega_{\mathrm{EA}}})|_{\Omega_{\mathrm{EC}}}$ where the probability of the event $\Omega_{\mathrm{EC}}$ with respect to the state $\sigma^\mathrm{QKD}|_{\Omega_{\mathrm{EA}}}$ is given by $\Pr_{\sigma^\mathrm{QKD}}[\Omega_{\mathrm{EC}} | \Omega_{\mathrm{EA}}]$.
Now we use Lemma B.5 of \cite{dupuis2016entropy} to show that for $a\in(1,2)$
\begin{align}
    &H^{\epsilon/16}_{\min}(O_1^n|S_1^nE)_{\sigma^\mathrm{QKD}|_{ \Omega_{\mathrm{NA}}}} 
    \nonumber\\
    &\geq H_{a}^{\uparrow}(O_1^n|S_1^nE)_{\sigma^\mathrm{QKD}|_{ \Omega_{\mathrm{NA}}}} - \frac{\Gamma(\epsilon/16)}{a-1} \\
    &\geq H_{a}^{\uparrow}(O_1^n|S_1^nE)_{\sigma^\mathrm{QKD}|_{ \Omega_{\mathrm{EA}}}} - \frac{\Gamma(\epsilon/16)}{a-1}\nonumber\\
    &\quad - \frac{a}{a-1} \log\left(\frac{1}{\Pr_{\sigma^\mathrm{QKD}}[\Omega_{\mathrm{EC}} | \Omega_{\mathrm{EA}}]}\right).\label{eq23}
\end{align}

In order to apply the EAT, we now consider the EAT process, which results in the same marginal state $\sigma^\mathrm{QKD}_{O_1^nS_1^nE}|_{ \Omega_{\mathrm{EA}}}$, hence the same value for $H_{a}^{\uparrow}(O_1^n|S_1^nE)_{\sigma^\mathrm{QKD}|_{ \Omega_{\mathrm{EA}}}}$,  as the hypothetical QKD protocol. The EAT process will be closely related to the  hypothetical QKD protocol, however it will not include the output $\hat{X}_i$, which is not necessary in this context. Also, the EAT process will not include an error correction or privacy amplification protocol.

We begin by defining our EAT channels. To that purpose, we take the channel $\mc{M}^\mathrm{QKD}$ defined in eq. (\ref{def:M123}); however we omit the output of $\hat{X}_i$, resulting in a channel 
\begin{equation}\label{def:M123tilde}
\mc{M}^\text{EAT}:A_iB_i \to O_iC_iS_i,
\end{equation}
which performs steps (1) - (3) of the hypothetical QKD protocol, but in the end does not output Alice's key system $\hat{X}_i$. It is easy to see that $C_i$ can be obtained by readout of classical information contained in $O_i$ and $S_i$. As it only contains discretised information, $O_i$ is finite dimensional. Further defining $Q_{i}:=A_{i+1}^nB_{i+1}^n$, we can now define a channel $\mc{M}^\text{EAT}_i:Q_{i-1}\to Q_{i} O_iS_iC_i$ by
\begin{equation}\label{eq:EATchan}
\mc{M}^\text{EAT}_i:=\id_{Q_{i}}\otimes\mc{M}^\text{EAT}_{A_{i}B_{i}\to O_iS_iC_i}.
\end{equation}
In order to apply the EAT, we still have to show that the Markov condition 
\begin{equation}\label{eq:Markov}
O_1^{i-1}\leftrightarrow S_1^{i-1}E\leftrightarrow S_i
\end{equation}
or, equivalently $I(O_1^{i-1}:S_i|S_1^{i-1}E)=0$, is fulfilled for all $i\in\{1,..,n\}$. In our case this holds trivially as $S_i=R_i$ is obtained by a local random number generator, which is used independently in each round. 

We can now define our EAT process as a concatenation of EAT channels $\mc{M}^\text{EAT}_n\circ\cdots\circ\mc{M}^\text{EAT}_1$, yielding the following state,
\begin{align}
&\sigma^\text{EAT}_{O_1^nS_1^nC_1^nE}\nonumber\\
&=\id_E\otimes\(\mc{M}^\text{EAT}_n\circ\cdots\circ\mc{M}^\text{EAT}_1\)\({\Psi}_{A_1^nB_1^nE}\)\label{eq_28}\\
&=\id_E\otimes{\mc{M}^\text{EAT}}^{\otimes n}\({\Psi}_{A_1^nB_1^nE}\)\\
&=\Tr_{\hat{X}_1^n}\left[\id_E\otimes{\mc{M}^\mathrm{QKD}}^{\otimes n}\({\Psi}_{A_1^nB_1^nE}\)\right].
\end{align}
The EAT process then concludes with Alice and Bob using $C_1^n$ to perform the tomography test as well as parameter estimation. This is done in the same way as in the hypothetical QKD protocol. Consequently it holds,
\begin{align}
\sigma^\mathrm{QKD}_{O_1^nS_1^nE}|_{\Omega_{\mathrm{EA}}}&=\sigma^\text{EAT}_{O_1^nS_1^nE}|_{\Omega_{\mathrm{EA}}},\label{eq_32}\\
H_{a}^{\uparrow}(O_1^n|S_1^nE)_{\sigma^\mathrm{QKD}|_{\Omega_{\mathrm{EA}}}}&=H_{a}^{\uparrow}(O_1^n|S_1^nE)_{\sigma^\text{EAT}|_{\Omega_{\mathrm{EA}}}}.\label{eq:33}
\end{align}

Hence, it will be sufficient to lower bound the r.h.s. of eq. (\ref{eq:33}) using the EAT. We can now define a \emph{min-tradeoff function} as a function  $f:\mc{P}_{\mc{C}}\to\mathbb{R}$ such that for all $i=1,...,n$ it holds

\begin{equation}\label{def:MinTrade}
f(p)\leq \inf_{\ket{\rho}\in\Sigma_i( p)}H(O_i|S_i\tilde{E})_{\rho^{\text{EAT},i}},
\end{equation}
where $\tilde{E}$ can be chosen isomorphic to $Q_{i-1}$, and we have defined 
\begin{equation}
\Sigma_i(p)=\left\{\ket{\rho}_{Q_{i-1}\tilde{E}}\in\mc{H}_{Q_{i-1}\tilde{E}}: \bra{c}\rho^{\text{EAT},i}_{C_i}\ket{c}\equiv p(c) \right\},
\end{equation}
for a state
\begin{align}
\rho^{\text{EAT},i}_{O_iS_iC_iQ_i\tilde{E}}&=\id_{\tilde{E}}\otimes\mc{M}^\text{EAT}_i(\rho_{Q_{i-1}\tilde{E}})\nonumber\\
&=\id_{Q_i\tilde{E}}\otimes\mc{M}^\text{EAT}_{A_iB_i\to O_iS_iC_i}\(\rho_{A_iB_iQ_i\tilde{E}}\).
\end{align}
Here $\equiv$ stands for equality for all $c\in\mc{C}$. Note that $\ket{\rho}$ can be chosen pure by strong subadditivity as remarked in \cite{dupuis2016entropy}.

In the following we will consider the case where $f(p)=f^\mathrm{PE}(p)+f^\mathrm{tom}(p)+\mathrm{const}$, with affine functions $f^\mathrm{PE}$ and $f^\mathrm{tom}$ of the form given by eqs. (\ref{eq:fPE},\ref{eq:ftom}), respectively.  In that case, it holds for all $c_1^n\in\Omega_\mathrm{EA}$ that $f(\freq_{c_1^n})\geq f(p_0)-\delta^\mathrm{tol}$. We can then formulate the entropy accumulation theorem, given by \cite[Proposition V.3]{dupuis2019entropy}, in the following way: 

\begin{proposition}\cite{dupuis2019entropy}\label{EAT2alpha}
Let $n\in\mathbb{N}$. Let $p_0$ be given by eqs. (\ref{eq:p0PE},\ref{eq:p0tom}). Let $\Omega_{\mathrm{EA}}$ be the event defined by eqs. (\ref{def:OmegaEA1}-\ref{def:OmegaEA2}) for some $\delta^\mathrm{tol}_\mathrm{PE},\delta^\mathrm{tol}_\mathrm{tom}>0$, $\delta^\mathrm{tol}=\delta^\mathrm{tol}_\mathrm{PE}+\delta^\mathrm{tol}_\mathrm{tom}$, and an affine  min-tradeoff function $f$ such that $f(p)=f^\mathrm{PE}(p)+f^\mathrm{tom}(p)+\mathrm{const}$. Then for $a\in(1,2)$, and a set of registers $O_1^n S_1^n E$ fulfilling the Markov chain \eqref{eq:Markov},
\begin{align}
&H^{\uparrow}_{a}(O_1^n|S_1^nE)_{\sigma^\mathrm{EAT}|_{ \Omega_{\mathrm{EA}}}}\nonumber\\
&\geq nf(p_0)-n\(\delta^\mathrm{tol}+\frac{(a-1)\ln2}{2}V^2\)\nonumber\\
&\quad-\frac{a}{a-1}\log\frac{1}{\Pr_{\sigma^\mathrm{EAT}}[\Omega_{\mathrm{EA}}]}-n(a-1)^2K_a,\label{eq39}
\end{align}
where we have defined
\begin{align}
V&=\sqrt{\mathrm{Var}(f)+2}+\log(2d_O^2+1),\label{eq:V}\\
K_a&=\frac{2^{(a-1)(2\log d_O+\max(f)-\min_\Sigma(f))}}{6(2-a)^3\ln2}\nonumber\\
&\quad\times\ln^3\(2^{2\log d_O+\max(f)-\min_\Sigma(f)}+e^2\),\label{eq:Ka}
\end{align}
where $\max(f)=\max_{p\in\mathcal{P}_{\mathcal{C}}}f(p)$ and $\min_\Sigma(f)=\min_{p:{\Sigma}\neq\oldemptyset}f(p)$ and $\mathrm{Var}(f)$ denotes the variance of $f$.

\end{proposition}

We will now use Proposition \ref{EAT2alpha} to show the soundness of the hypothetical protocol. For that purpose, we need the following Lemma, which formalises the intuition that in order to upper bound the probability of aborting after tomography, we have to choose the corresponding tolerance parameter large enough.

\begin{lemma}\label{lem:deltatolbound}
Let $n\in\mathbb{N}$ and $\epsilon^\mathrm{tom}\in(0,1)$. Let us assume it holds
   \begin{equation}\label{eq:deltatom}
    \delta^\mathrm{tol}_\mathrm{tom}\geq 2\sqrt{\log\(\frac{n}{\epsilon^\mathrm{tom}}\)\sum_{i=1}^{16}\frac{\gamma'_ic'^2_i}{n}}+\frac{3D'}{n}\log\frac{n}{\epsilon^\mathrm{tom}},
\end{equation}
where we have defined for, $i\in\{1,...,16\}$,
\begin{align}
&\gamma'_i:=\frac{\pi'_i(1-\sum_{j=1}^i\pi'_j)}{1-\sum_{j=1}^{i-1}\pi'_j},\\
&c'_i:=h'_i-\frac{\sum_{j=i+1}^{17}h'_j\pi'_j}{1-\sum_{j=1}^{i}\pi'_j},\\
&D':=\max_{i,j\in\{1,...,17\}}|h'_i-h'_j|,
\end{align}
for $\pi'_{x'+1}=p_0(\perp,\perp,x')$, and $h'_{x'+1}=h_{\perp,\perp,x'}$, for $x'=0,...,15$, as well as $\pi'_{17}=1-\sum_{i=1}^{16}\pi'_i$, and $h'_{17}=0$. Then it holds \begin{equation}
    \Pr_{\sigma^\mathrm{QKD}}[ \neg\Omega_{\mathrm{tom}}]\leq\epsilon^\mathrm{tom}.
\end{equation}
\end{lemma}

\begin{proof}
As the tomography is performed entirely within Alice's lab, with no influence of Eve or the noisy channel, we can restrict our attention to an honest implementation of the protocol, i.e. $\Pr_{\sigma^\mathrm{QKD}}[ \neg\Omega_{\mathrm{tom}}]=\Pr_{\mathrm{hon}}[ \neg\Omega_{\mathrm{tom}}]$. Let us assume an honest application gives us the distribution
\begin{align}
     p_0(\perp,\perp,x')&=(1-p^\text{key})\tilde p_0(\perp,\perp,x')\nonumber\\
     &=p^\text{tom}p_0^\text{tom}(x'),
\end{align}
for  $x'\in\{0,...,15\}$. Let us further assume that Alice and Bob observe some frequency distribution $\freq_{c^n_1}$. Recalling the definition of the event $\Omega_\mathrm{tom}$, we note that it holds 
\begin{align}
     &\Pr_\mathrm{hon}[\Omega_\mathrm{tom}]\nonumber\\
     &\geq\Pr_\mathrm{hon}\left[\left|f^\mathrm{tom}(\freq_{c_1^n})-f^\mathrm{tom}(p_0)\right|\leq\delta_\mathrm{tom}^\mathrm{tol}\right].
\end{align}
An honest implementation of the protocol corresponds to $n$ independent multinoulli trials with parameter $p_0$. In order to provide lower bounds we can therefore make use of a concentration result provided by Proposition 2 of \cite{agrawal2017posterior}. Namely, it holds with probability $1-\epsilon^\mathrm{tom}$ that
\begin{align}
    &\left|f^\mathrm{tom}(\freq_{c_1^n})-f^\mathrm{tom}(p_0)\right|=\left|(\hat{\pi}'-\pi')^Th'\right|\nonumber\\
    &\leq 2\sqrt{\log\(\frac{n}{\epsilon^\mathrm{tom}}\)\sum_{i=1}^{16}\frac{\gamma'_ic'^2_i}{n}}+\frac{3D'}{n}\log\frac{n}{\epsilon^\mathrm{tom}},
\end{align}
where $\hat{\pi}'_{x'+1}=\freq_{c_1^n}(\perp,\perp,x')$ for $x'=0,...,15$, as well as  $\hat{\pi}'_{17}=1-\sum_{i=1}^{16}\hat{\pi}'_i$. Hence, if we choose the tolerance parameter $\delta^\mathrm{tol}_\mathrm{tom}$ fulfilling eq. (\ref{eq:deltatom}), we obtain the desired bound
\begin{equation}
\Pr_\mathrm{hon}\left[\left|f^\mathrm{tom}(\freq_{c_1^n})-f^\mathrm{tom}(p_0)\right|\leq\delta_\mathrm{tom}^\mathrm{tol}\right]\geq1-\epsilon^\mathrm{tom},
\end{equation}
finishing the proof.

\end{proof}

In order to show the soundness of the physical QKD protocol, which does not include tomography of Alice's marginal system, we also need the following Lemma, which relates the smooth min entropies of the physical and hypothetical protocol.

\begin{lemma}\label{lem:HypPhys}
Let the smoothing parameters \\ $\epsilon\in\(0,1-\sqrt{2\Pr_{\sigma^\mathrm{QKD}}[\neg\Omega_\mathrm{tom} | \Omega_\mathrm{EC}\cap\Omega_\mathrm{PE}]}\)$ and  $\epsilon^\mathrm{phys}\in\(\epsilon+\sqrt{2\Pr_{\sigma^\mathrm{QKD}}[\neg\Omega_\mathrm{tom} | \Omega_\mathrm{EC}\cap\Omega_\mathrm{PE}]},1\)$. 
Then it holds
\begin{align}
    &H^{\epsilon^\mathrm{phys}}_{\min}(\hat{Z}_1^n|S_1^n\tilde{X}_1^nE)_{\sigma^\mathrm{phys,QKD}|_{\Omega^\mathrm{phys}_{\mathrm{NA}}}}\nonumber\\
    &\geq H^\epsilon_{\min}(\hat{Z}_1^n|S_1^n\tilde{X}_1^nE)_{\sigma^\mathrm{QKD}|_{\Omega_{\mathrm{NA}}}}.\label{eq68}
\end{align}
\end{lemma}

\begin{proof}
We begin by noting that Alice's tomography in the hypothetical protocol does not change the $\hat{Z}_1^nS_1^n\tilde{X}_1^nE$ subsystems of the final state, i.e. 
\begin{equation}\label{76}
\sigma_{\hat{Z}_1^nS_1^n\tilde{X}_1^nHH'\bar{X}_1^n\tilde{Y}_1^nE}^\mathrm{phys,QKD}=\sigma_{\hat{Z}_1^nS_1^n\tilde{X}_1^nHH'\bar{X}_1^n\tilde{Y}_1^nE}^\mathrm{QKD}.
\end{equation}
This is by non-signalling. As the events $\Omega_\mathrm{EC}$ and $\Omega_\mathrm{PE}$, defined by eqs. (\ref{eq:OmegaPE}) and (\ref{eq:OmegaEC}), respectively, depend only on the systems $HH'\hat{Z}_1^n\bar{X}_1^n$ and $\tilde{Y}_1^n\tilde{X}_1^n$, it also holds
\begin{align}
&\sigma_{\hat{Z}_1^nS_1^n\tilde{X}_1^nHH'\bar{X}_1^n\tilde{Y}_1^nE}^\mathrm{phys,QKD}\big|_{\Omega_\mathrm{EC}\cap\Omega_\mathrm{PE}}\nonumber\\
&=\sigma_{\hat{Z}_1^nS_1^n\tilde{X}_1^nHH'\bar{X}_1^n\tilde{Y}_1^nE}^\mathrm{QKD}\big|_{\Omega_\mathrm{EC}\cap\Omega_\mathrm{PE}}.\label{76a}
\end{align}
Using the triangle inequality, it follows that
\begin{align}
&\left\|\sigma_{\hat{Z}_1^nS_1^n\tilde{X}_1^nE}^\mathrm{QKD}|_{\Omega_\mathrm{EC}\cap\Omega_\mathrm{PE}\cap\Omega_\mathrm{tom}}-\sigma_{\hat{Z}_1^nS_1^n\tilde{X}_1^nE}^\mathrm{phys,QKD}\big|_{\Omega_\mathrm{EC}\cap\Omega_\mathrm{PE}}\right\|_1 \nonumber\\
&\leq
\left\|\sigma_{\hat{Z}_1^nS_1^n\tilde{X}_1^nE}^\mathrm{QKD}|_{\Omega_\mathrm{EC}\cap\Omega_\mathrm{PE}\cap\Omega_\mathrm{tom}}- \sigma_{\hat{Z}_1^nS_1^n\tilde{X}_1^nE}^\mathrm{QKD}|_{\Omega_\mathrm{EC}\cap\Omega_\mathrm{PE}} \right\|_1\nonumber\\
&\quad +\left\| \sigma_{\hat{Z}_1^nS_1^n\tilde{X}_1^nE}^\mathrm{QKD}|_{\Omega_\mathrm{EC}\cap\Omega_\mathrm{PE}}
-\sigma_{\hat{Z}_1^nS_1^n\tilde{X}_1^nE}^\mathrm{phys,QKD}\big|_{\Omega_\mathrm{EC}\cap\Omega_\mathrm{PE}}\right\|_1\nonumber
\\
&\leq 2\Pr_{\sigma^\mathrm{QKD}}[\neg\Omega_\mathrm{tom} | \Omega_\mathrm{EC}\cap\Omega_\mathrm{PE}].
\end{align}

Hence by eq. (\ref{Fuchs}) it holds 
\begin{align}
&P\left(\sigma_{\hat{Z}_1^nS_1^n\tilde{X}_1^nE}^\mathrm{QKD}|_{\Omega_\mathrm{EC}\cap\Omega_\mathrm{PE}\cap\Omega_\mathrm{tom}},\sigma_{\hat{Z}_1^nS_1^n\tilde{X}_1^nE}^\mathrm{phys,QKD}|_{\Omega_\mathrm{EC}\cap\Omega_\mathrm{PE}}\right)\nonumber\\
&\leq\sqrt{ 2\Pr_{\sigma^\mathrm{QKD}}[\neg\Omega_\mathrm{tom} | \Omega_\mathrm{EC}\cap\Omega_\mathrm{PE}]}.
\end{align}
Further, by the triangle inequality, the $\epsilon$-ball (in terms of purified distance) around $\sigma_{\hat{Z}_1^nS_1^n\tilde{X}_1^nE}^\mathrm{QKD}|_{\Omega_{\mathrm{NA}}}$ is contained in the $\epsilon^\mathrm{phys}$-ball around $\sigma_{\hat{Z}_1^nS_1^n\tilde{X}_1^nE}^\mathrm{QKD}|_{\Omega^\mathrm{phys}_{\mathrm{NA}}}$, implying eq. (\ref{eq68}).

\end{proof}

We are now ready to prove the soundness physical QKD protocol, which is our main result.

\begin{theorem}{\bf(Soundness)}\label{FinteSizePhys}
Let  $n\in\mathbb{N}$. Let $\epsilon^\mathrm{phys}_{\mathrm{NA}},\epsilon^\mathrm{tom},\epsilon_\mathrm{EC}\in(0,1)$ such that $\epsilon^\mathrm{tom}<\frac{1}{2}\epsilon^\mathrm{phys}_{\mathrm{NA}}$. Let $\epsilon\in\(0,1-\sqrt{2\epsilon^\mathrm{tom}/\epsilon^\mathrm{phys}_{\mathrm{NA}}}\)$, and define $\epsilon^\mathrm{phys}=\epsilon+\sqrt{2\epsilon^\mathrm{tom}/\epsilon^\mathrm{phys}_{\mathrm{NA}}}$. Let $0\leq  p^\mathrm{key},p^\mathrm{PE},p^\mathrm{tom}\leq1$ such that $p^\mathrm{key}+p^\mathrm{PE}+p^\mathrm{tom}=1$. Let $f$ be an affine min-tradeoff function of the form $f(p)=f^\mathrm{PE}(p)+f^\mathrm{tom}(p)+\mathrm{const}$. Let $p_0$ be given by eqs. (\ref{eq:p0PE},\ref{eq:p0tom}). Let $\delta^\mathrm{tol}_\mathrm{PE}>0$ and define 
\begin{equation}\label{eq:deltatom2}
    \delta^\mathrm{tol}_\mathrm{tom}= 2\sqrt{\log\(\frac{n}{\epsilon^\mathrm{tom}}\)\sum_{i=1}^{16}\frac{\gamma'_ic'^2_i}{n}}+\frac{3D'}{n}\log\frac{n}{\epsilon^\mathrm{tom}}.
\end{equation}
Let further $\Omega_\mathrm{PE}$, $\Omega_\mathrm{tom}$, $\Omega_\mathrm{EC}$, and $\Omega_\mathrm{NA}$ be defined by eqs. (\ref{eq:OmegaPE}-\ref{eq:OmegaNonAbort}). Let $\text{leak}_\text{EC}$ be the amount of information leaked during error correction. Then, if $\Pr_{\sigma^\mathrm{phys,QKD}}[\Omega_{\mathrm{PE}}\cap\Omega_{\mathrm{EC}}]\geq \epsilon^\mathrm{phys}_{\mathrm{NA}}$, 
for any $a\in(1,2)$, the physical  QKD protocol provides an $3\epsilon^\mathrm{phys}+\epsilon_\mathrm{EC}$-sound key at rate $r^\mathrm{phys}=\ell/n$ with
\begin{align}
  & r^\mathrm{phys}|_{\Omega^\mathrm{phys}_{\mathrm{NA}}}\nonumber\\
   &\geq f(p_0) -\delta^\mathrm{tol}_\mathrm{PE}-\delta^\mathrm{tol}_\mathrm{tom}-\frac{(a-1)\ln2}{2}V^2\nonumber\\
   &-(a-1)^2K_a-p^\mathrm{PE}\log{5}\nonumber\\
   &-(1-p^\mathrm{key})\log(17(m+1))\nonumber\\
    &-\frac{1}{\sqrt{n}}\left[\sqrt{\frac{1}{2}\ln\frac{32}{\epsilon^2 ( \epsilon^\mathrm{phys}_{\mathrm{NA}}-\epsilon^\mathrm{tom})}}\log{5}\right.\nonumber\\
    &\left.+\sqrt{\frac{1}{2}\ln\frac{512}{\epsilon^2 ( \epsilon^\mathrm{phys}_{\mathrm{NA}}-\epsilon^\mathrm{tom})}}\log{(17(m+1))}\right]\nonumber\\
    &-\frac{1}{n}\left[\frac{\Gamma(\epsilon/16)}{a-1}+\frac{a}{a-1}\log\frac{1}{ \epsilon^\mathrm{phys}_{\mathrm{NA}}-\epsilon^\mathrm{tom}}\right.\nonumber\\
    &\left.+\mathrm{leak}_\mathrm{EC}+2\log{\frac{1}{\epsilon^\mathrm{phys}}}+2\Gamma(\epsilon/4)+3\Gamma(\epsilon/16)\right]\label{eq:FiniteKeyRate}.
\end{align}

\end{theorem}

\begin{proof}
We begin by lower bounding $H^\epsilon_{\min}(\hat{Z}_1^n|S_1^n\tilde{X}_1^nE)_{\sigma^\mathrm{QKD}|_{{ \Omega_{\mathrm{NA}}}}}$ using eqs. (\ref{eq:20} - \ref{eq23}). We then note that by eq. (\ref{eq_32}) it holds $H_a^\uparrow(O_1^n|S_1^nE)_{\sigma^\mathrm{QKD}|_{\Omega_{\mathrm{EA}}}}=H_a^\uparrow(O_1^n|S_1^nE)_{\sigma^\text{EAT}|_{\Omega_{\mathrm{EA}}}}$, allowing us to apply Proposition \ref{EAT2alpha}. Since  $\Pr_{\sigma^\mathrm{QKD}}[\Omega_{\mathrm{EC}}|\Omega_{\mathrm{EA}}]\Pr_{\sigma^\mathrm{QKD}}[\Omega_{\mathrm{EA}}]=\Pr_{\sigma^\mathrm{QKD}}[\Omega_{\mathrm{NA}}]$, the terms $\frac{a}{a-1}\log\frac{1}{ \Pr_{\sigma^\mathrm{QKD}}[\Omega_{\mathrm{EC}}|\Omega_{\mathrm{EA}}]}$ and $\frac{a}{a-1}\log\frac{1}{ \Pr_{\sigma^\mathrm{QKD}}[\Omega_{\mathrm{EA}}]}$ in eqs (\ref{eq23}) and (\ref{eq39}), can be merged, resulting in a term that  dependents only on $\Pr_{\sigma^\mathrm{QKD}}[\Omega_{\mathrm{NA}}]$. Formally, we obtain 

\begin{align}
    &H^\epsilon_{\min}(\hat{Z}_1^n|S_1^n\tilde{X}_1^nE)_{\sigma^\mathrm{QKD}|_{\Omega_{\mathrm{NA}}}} \nonumber\\
    &\geq n\left[ f(p_0) -\delta^\mathrm{tol}_\mathrm{PE}-\delta^\mathrm{tol}_\mathrm{tom}-\frac{(a-1)\ln2}{2}V^2\right.\nonumber\\
    &-(a-1)^2K_a-p^\mathrm{PE}\log{5}\nonumber\\
    &\left.-(1-p^\mathrm{key})\log(17(m+1))\right]\nonumber\\
    &-\sqrt{n}\left[\sqrt{\frac{1}{2}\ln\frac{32}{\epsilon^2 {\Pr_{\sigma^\mathrm{QKD}}[\Omega_{\mathrm{NA}}]}}}\log{5}\right.\nonumber\\
    &\left.+\sqrt{\frac{1}{2}\ln\frac{512}{\epsilon^2 {\Pr_{\sigma^\mathrm{QKD}}[\Omega_{\mathrm{NA}}]}}}\log{(17(m+1))}\right]\nonumber\\
    &-\frac{\Gamma(\epsilon/16)}{a-1}-\frac{a}{a-1}\log\frac{1}{\Pr_{\sigma^\mathrm{QKD}}[\Omega_{\mathrm{NA}}]}\nonumber\\
&-2\Gamma(\epsilon/4)-3\Gamma(\epsilon/16).
\end{align}
Further, it holds
\begin{align}
&\Pr_{\sigma^\mathrm{QKD}}[\Omega_{\mathrm{NA}}]\nonumber\\
&=  \Pr_{\sigma^\mathrm{QKD}}[\Omega_{\mathrm{PE}}\cap\Omega_{\mathrm{EC}}]\nonumber \\
&\quad  -\Pr_{\sigma^\mathrm{QKD}}[\Omega_{\mathrm{PE}}\cap\Omega_{\mathrm{EC}}\cap \neg\Omega_{\mathrm{tom}}]\nonumber\\
&\geq\Pr_{\sigma^\mathrm{QKD}}[\Omega_{\mathrm{PE}}\cap\Omega_{\mathrm{EC}}]- \Pr_{\sigma^\mathrm{QKD}}[ \neg\Omega_{\mathrm{tom}}].
\end{align}
By eq. (\ref{eq:deltatom2}) and Lemma \ref{lem:deltatolbound}, we can bound $\Pr_{\sigma^\mathrm{QKD}}[\neg\Omega_\mathrm{tom}]\leq\epsilon^\mathrm{tom}$. Further, by eq. (\ref{76}), it holds that $\Pr_{\sigma^\mathrm{phys,QKD}}[\Omega_{\mathrm{PE}}\cap\Omega_{\mathrm{EC}}]=\Pr_{\sigma^\mathrm{QKD}}[\Omega_{\mathrm{PE}}\cap\Omega_{\mathrm{EC}}]$. Hence, by assumption, it holds
\begin{align}
   &\Pr_{\sigma^\mathrm{QKD}}[\Omega_{\mathrm{NA}}]
   \geq \epsilon^\mathrm{phys}_{\mathrm{NA}}-\epsilon^\mathrm{tom}.
\end{align}
Now we can apply Lemma \ref{lem:HypPhys} to show that for our choice of $\epsilon$, and $\epsilon^\mathrm{phys}=\epsilon+\sqrt{2\epsilon^\mathrm{tom}/\epsilon^\mathrm{phys}_{\mathrm{NA}}}$, it holds
\begin{align}
 &H^{\epsilon^\mathrm{phys}}_{\min}(\hat{Z}_1^n|S_1^n\tilde{X}_1^nE)_{\sigma^\mathrm{phys,QKD}|_{\Omega^\mathrm{phys}_{\mathrm{NA}}}}\nonumber\\
 &\geq  H^\epsilon_{\min}(\hat{Z}_1^n|S_1^n\tilde{X}_1^nE)_{\sigma^\mathrm{QKD}|_{\Omega_{\mathrm{NA}}}}
\end{align}
and apply Proposition \ref{renner}, finishing the proof.
\end{proof}

\subsection{Completeness}\label{sec:compl}

In this section we show that the physical QKD protocol is complete, i.e. we provide a lower bound on the probability $\Pr_\mathrm{hon}[\Omega^\mathrm{phys}_{\mathrm{NA}}]$ of an honest application not aborting.

\begin{theorem}\label{theo:complete}{\bf(Completeness)}
Let $\epsilon_\mathrm{PE}^c\in(0,1)$. Let $\Omega_\mathrm{PE}$ be as defined in eq. (\ref{eq:OmegaPE}), with 
\begin{equation}\label{eq:deltaPE}
    \delta^\mathrm{tol}_\mathrm{PE}= 2\sqrt{\log\(\frac{n}{\epsilon_\mathrm{PE}^c}\)\sum_{i=1}^{4m}\frac{\gamma_ic^2_i}{n}}+\frac{3D}{n}\log\frac{n}{\epsilon_\mathrm{PE}^c},
\end{equation}
where we have defined an ordering $(x,z,\perp)\to i$ by $(0,0,\perp)\to1,(0,1,\perp)\to2,...,(0,m-1,\perp)\to m,(1,0,\perp)\to m+1,(1,1,\perp)\to m+2,...,(3,m-1,\perp)\to4m$. We then set $\pi_i=p_0(x,z,\perp)$, $\hat{\pi}_i=\freq_{c_1^n}(x,z,\perp)$ $h_i=h_{x,z,\perp}$ for $i=1,...,4m$, as well as $\pi_{4m+1}:=1-\sum_{i=1}^{4m}\pi_i$, $\hat{\pi}_{4m+1}:=1-\sum_{i=1}^{4m}\hat{\pi}_i$ and $h_{4m+1}=0$. Let us further define
    \begin{align}
    &\gamma_i:=\frac{\pi_i\(1-\sum_{j=1}^i\pi_j\)}{1-\sum_{j=1}^{i-1}\pi_j},\label{eq:79a}\\
    &c_i:=h_i-\frac{\sum_{j=i+1}^{4m+1}h_j\pi_j}{1-\sum_{j=1}^{i}\pi_j},\label{eq:79b}
    \end{align}
    for $i=1,...,4m$, as well as 
    $D:=\max_{i,j\in\{1,...,4m+1\}}|h_i-h_j|$.
    Let further $\epsilon_\mathrm{EC}^c\in(0,1)$ be a suitable completeness parameter for the error correction protocol used. Then the physical QKD protocol is $\epsilon_\mathrm{PE}^c+\epsilon_\mathrm{EC}^c$-complete, i.e. $\Pr_\mathrm{hon}[\Omega^\mathrm{phys}_\mathrm{NA}]\geq1-\epsilon_\mathrm{PE}^c -\epsilon_\mathrm{EC}^c$.
\end{theorem}

\begin{proof}
We consider the following honest implementation: We apply the physical QKD protocol, as described in Section \ref{sec:Protocol}, where the noisy channel $\mc{N}_{{A'}_1^n\to B_1^n}$ is given by $n$ iid uses of a phase-invariant Gaussian channel with transmittance $\eta$ and excess noise $\xi$, however without an attack by Eve. By simulating this channel we obtain a distribution $p_0^\mathrm{sim}(x,z)$, which depends on $\eta$ and $\xi$ and is given by eq. (\ref{eq:doubleint}), and set $p_0(x,z,\perp)=p^\mathrm{PE}p_0^\mathrm{sim}(x,z)$ for $x\in\{0,...,3\}$, $z\in\{0,...,m-1\}$. We note that the protocol can abort after parameter estimation or error correction. By the union bound it holds
\begin{equation}
    1-\Pr_\mathrm{hon}[\Omega^\mathrm{phys}_{\mathrm{NA}}]\leq1- \Pr_\mathrm{hon}[\Omega_\mathrm{PE}]+1-\Pr_\mathrm{hon}[\Omega_\mathrm{EC}].
\end{equation}

We begin by considering abortion after parameter estimation. Let us assume an honest application gives us $p_0(x,z,\perp)$,
for $x\in\{0,...,3\}$, $z\in\{0,...,m-1\}$ according to eq. (\ref{eq:p0PE}).
Let us further assume that Alice and Bob observe some frequency distribution $\freq_{c^n_1}$. Recalling the definition of the event $\Omega_\mathrm{PE}$, we note that it holds 
\begin{align}
    &\Pr_\mathrm{hon}[\Omega_\mathrm{PE}] \nonumber
    \\
    &\geq\Pr_\mathrm{hon}\left[\left|f^\mathrm{PE}(\freq_{c_1^n})-f^\mathrm{PE}(p_0)\right|\leq\delta_\mathrm{PE}^\mathrm{tol}\right].
\end{align}
An honest implementation of the protocol corresponds to $n$ independent multinoulli trials with parameter $p_0$. In order to provide lower bounds we can again make use of the concentration result provided by Proposition 2 of~\cite{agrawal2017posterior}.

Let now $\epsilon_\mathrm{PE}^c\in(0,1)$. By Proposition 2 of \cite{agrawal2017posterior}, it then holds with probability $1-\epsilon_\mathrm{PE}^c$ that 
\begin{align}
    &\left|f^\mathrm{PE}(\freq_{c_1^n})-f^\mathrm{PE}(p_0)\right|=\left|(\hat{\pi}-\pi)^Th\right|\nonumber\\
    &\leq 2\sqrt{\log\(\frac{n}{\epsilon_\mathrm{PE}^c}\)\sum_{i=1}^{4m}\frac{\gamma_ic^2_i}{n}}+\frac{3D}{n}\log\frac{n}{\epsilon_\mathrm{PE}^c}.
\end{align}
Hence, if we choose the tolerance parameter $\delta^\mathrm{tol}_\mathrm{PE}$ as in eq. (\ref{eq:deltaPE}), we obtain the desired completeness bound
\begin{equation}
\Pr_\mathrm{hon}\left[\left|f^\mathrm{PE}(\freq_{c_1^n})-f^\mathrm{PE}(p_0)\right|\leq\delta_\mathrm{PE}^\mathrm{tol}\right]\geq1-\epsilon_\mathrm{PE}^c.
\end{equation}

Finally, we have to consider the error correction. Let $\epsilon_\mathrm{EC}^c\in(0,1)$, such that $1-\Pr_\mathrm{hon}[\Omega_\mathrm{EC}]\leq \epsilon_\mathrm{EC}^c$, i.e. the error correction is assumed to abort with probability at most $\epsilon_\mathrm{EC}^c$, finishing the proof.
\end{proof}

For a given empirical distribution $p_0$ and a suitable choice of a min-tradeoff function, Theorems \ref{FinteSizePhys} and \ref{theo:complete} combined show the security of the finite round physical QKD protocol. In order to obtain the best finite size key rate, for given $n\in\mathbb{N}$, as well as some choice for the parameters $\epsilon^\mathrm{phys}_{\mathrm{NA}},\epsilon^\mathrm{tom},\epsilon_\text{EC},\epsilon_\mathrm{EC}^c,\epsilon_\mathrm{PE}^c\in(0,1)$, such that $\epsilon^\mathrm{tom} \ll \frac{1}{2}\epsilon^\mathrm{phys}_{\mathrm{NA}}$, as well as $\epsilon\in\(0,1-\sqrt{2\epsilon^\mathrm{tom}/\epsilon^\mathrm{phys}_{\mathrm{NA}}}\)$ and $\epsilon^\mathrm{phys}=\epsilon+\sqrt{2\epsilon^\mathrm{tom}/\epsilon^\mathrm{phys}_{\mathrm{NA}}}$, we can set the tolerance parameters as in eqs. (\ref{eq:deltatom2},\ref{eq:deltaPE}), and maximise the key rate given by (\ref{eq:FiniteKeyRate}) over probabilities $0\leq  p^\text{key},p^\text{PE},p^\text{tom}\leq1$ such that $p^\text{key}+p^\text{PE}+p^\text{tom}=1$, 
and over $a\in(1,2)$. We note that in order to get a non-trivial result, we will have to choose $\epsilon_\mathrm{PE}^c$ and $\epsilon_\mathrm{EC}^c$ such that the success probability of the honest implementation meets the threshold $\epsilon^\mathrm{phys}_{\mathrm{NA}}$ used in Theorem \ref{FinteSizePhys}, i.e. we need $1-\epsilon_\mathrm{PE}^c-\epsilon_\mathrm{EC}^c>\epsilon^\mathrm{phys}_{\mathrm{NA}}$.

\subsection{The Min-Tradeoff Function}\label{sec:MinTrade}

The main task now is to find a min-tradeoff function $f$ that provides a non-trivial bound for our protocol. As we will choose the number of key rounds to be significantly larger than the number of test rounds (i.e. rounds used for parameter estimation or tomography), it will be convenient to use the infrequent sampling framework introduced in \cite{dupuis2019entropy}, in which the statistical analysis only includes outputs in test rounds. To that purpose, we divide $\mc{M}^\text{EAT}_i$ into a key part, incorporating its action in key rounds; and a test part, incorporating its action in parameter estimation and tomography rounds, $\mc{M}^\text{EAT,key}_i:Q_{i-1}\to Q_{i} O_iS_i$ and $\mc{M}^\text{EAT,test}_i:Q_{i-1}\to Q_{i} O_iS_iC_i$, such that
\begin{align}
\mc{M}^\text{EAT}_{i}(\cdot)&=p^\text{key}\mc{M}^\text{EAT,key}_{i}(\cdot)\otimes\pro{\perp}_{C_i}\nonumber\\
&\quad +(1-p^\text{key})\mc{M}^\text{EAT,test}_{i}(\cdot).
\end{align}
Let us now define a \emph{crossover min-tradeoff function} \cite{dupuis2019entropy} as a function  $g:\mc{P}_{\tilde{\mc{C}}}\to\mathbb{R}$ such that for all $i=1,...,n$ and $\tilde{p}\in\mc{P}_{\tilde{\mc{C}}}$ it holds

\begin{equation}\label{eq:MinTrade}
g(\tilde{p})\leq \inf_{\ket{\rho}\in\tilde\Sigma_i(\tilde p)}H(O_i|S_i\tilde{E})_{\rho^{\text{EAT},i}},
\end{equation}
where $\tilde{E}$ can be chosen isomorphic to $Q_{i-1}$, and we have defined
\begin{align}
\tilde\Sigma_i(\tilde p)=&\left\{\ket{\rho}_{Q_{i-1}\tilde{E}}\in\mc{H}_{Q_{i-1}\tilde{E}}:  \right. \nonumber \\
&\quad \left.\bra{c}\rho^{\text{EAT,test},i}_{C_i}\ket{c}\equiv \tilde p(c) \right\},
\end{align}
for states
\begin{align}
&\rho^{\text{EAT,test},i}_{O_iS_iC_iQ_i\tilde{E}}=\id_{\tilde{E}}\otimes\mc{M}^\text{EAT,test}_i(\rho_{Q_{i-1}\tilde{E}})\nonumber\\
&=\id_{Q_i\tilde{E}}\otimes\mc{M}^\text{EAT,test}_{A_iB_i\to O_iS_iC_i}\(\rho_{A_iB_iQ_i\tilde{E}}\).
\end{align}
Further, it holds for all $i=1,...,n$,
\begin{equation}\label{eq:35}
 \inf_{\ket{\rho}\in\tilde\Sigma_i(\tilde p)}H(O_i|S_i\tilde{E})_{\rho^{\text{EAT},i}}\geq \inf_{\substack{\mc{H}_{\hat{E}}\simeq\mc{H}_{Q_1}\\\ket{\rho}\in\tilde\Sigma_{\hat{E}}(\tilde p)}}H(O|S\hat{E})_{\rho^\text{EAT}},
\end{equation}
where we have defined the states $\rho^\text{EAT}_{OSC\hat{E}}=\id_{\hat{E}}\otimes\mc{M}^\text{EAT}_{AB\to OSC}\(\rho_{AB\hat{E}}\)$ and $\rho^\text{EAT,test}_{OSC\hat{E}}=\id_{\hat{E}}\otimes\mc{M}^\text{EAT,test}_{AB\to OSC}\(\rho_{AB\hat{E}}\)$, as well as the set $\tilde\Sigma_{\hat{E}}(\tilde p)=\left\{\ket{\rho}_{AB\hat{E}}\in \mc{H}_{AB\hat{E}}:\bra{c}\rho^\text{EAT,test}_{C}\ket{c}\equiv \tilde p(c) \right\}$. We can therefore relax the problem to finding a function $g:\mc{P}_{\tilde{\mc{C}}}\to\mathbb{R}$ such that 
\begin{equation}
g(\tilde p)\leq \inf_{\substack{\mc{H}_{\hat{E}}\simeq\mc{H}_{Q_1}\\\ket{\rho}\in\tilde\Sigma_{\hat{E}}( \tilde p)}}H(O|S\hat{E})_{\rho^\text{EAT}}.
\end{equation}
According to Lemma V.5 of \cite{dupuis2019entropy}, we translate our crossover min-tradeoff function $g$ into a min-tradeoff function $f$ via the definition
\begin{align}
&f(\delta_c)=\max(g) +\frac{g(\delta_c)-\max(g)}{1-p^\text{key}}\;\forall c\in\tilde{\mc{C}},\label{def:f1}\\
&f(\delta_{(\perp,\perp,\perp)})=\max(g),\label{def:f2}
\end{align}
where $\delta_c$ denotes the distribution that equals $1$ for $c$ and $0$ everywhere else. Further, $\max(g)=\max_{\tilde p\in\mc{P}_{\tilde{\mc{C}}}}g(\tilde p)$ and $\min(g)=\min_{\tilde p\in \mc{P}_{\tilde{\mc{C}}}}g(\tilde p)$. If $p$ is of the form $p(c)=(1-p^\text{key})\tilde p(c)$ for $c\in\tilde{\mathcal{C}}$ and $p(\perp,\perp,\perp)=p^\text{key}$, it holds $f((1-p^\text{key})\tilde p)=g(\tilde p)$ for all $\tilde p\in\mc{P}_{\tilde{\mc{C}}}$. Further it holds
\begin{align}
&\max(f)=\max(g),\label{e46}\\
&\min_\Sigma(f)\geq\min(g),\\
0\leq&\mathrm{Var}(f)\leq\frac{1}{1-p^\text{key}}\(\max(g)-\min(g)\)^2\label{e48}.
\end{align}
Hence we can upper bound the expressions in eqs. (\ref{eq:V},\ref{eq:Ka}) by

\begin{align}
V\leq\tilde V&=\sqrt{\frac{1}{1-p^\text{key}}\(\max(g)-\min(g)\)^2+2}\nonumber\\
&\quad +\log(2d_O^2+1),\label{eq:96}\\
K_a\leq\tilde K_a&=\frac{2^{(a-1)(2\log d_O+\max(g)-\min(g))}}{6(2-a)^3\ln2}\nonumber\\
&\quad\times\ln^3\(2^{2\log d_O+\max(g)-\min(g)}+e^2\).\label{eq:97}
\end{align}

In what follows, we provide a crossover minimum tradeoff function for our choice of EAT channels $\{\mc{M}^\text{EAT}_i\}_{i=1}^n$, lower bounding the r.h.s. of (\ref{eq:MinTrade}). 
We begin by noting that, by the chain rule for the von Neumann entropy \cite{wilde2013quantum}, it holds
\begin{align}
&\inf_{\substack{\mc{H}_{\hat{E}}\simeq\mc{H}_{Q_1}\\\ket{\rho}\in\Sigma_{\hat{E}}(\tilde p)}}H(O|S\hat{E})_{\rho^\text{EAT}}\nonumber\\
&\geq\inf_{\substack{\mc{H}_{\hat{E}}\simeq\mc{H}_{Q_1}\\\ket{\rho}\in\Sigma_{\hat{E}}(\tilde p)}}\( H(\hat{Z}|S\hat{E})_{\rho^\text{EAT}}+H(\tilde{Z}\tilde{X}X'|\hat{Z}S\hat{E})_{\rho^\text{EAT}}\)\label{eq:36}\\
&\geq \inf_{\substack{\mc{H}_{\hat{E}}\simeq\mc{H}_{Q_1}\\\ket{\rho}\in\Sigma_{\hat{E}}(\tilde p)}} H(\hat{Z}|S\hat{E})_{\rho^\text{EAT}},\label{eq:37}
\end{align}
where we have used that, as $\tilde{Z}\tilde{X}X'$ is classical, there cannot be any entanglement across the $\tilde{Z}\tilde{X}X':\hat{Z}S\hat{E}$ partition, hence the second term in (\ref{eq:36}) has to be non-negative. Let us now define $g:\mc{P}_{\tilde{\mc{C}}}\to\mathbb{R}$,
\begin{equation}\label{eq:MinTradeOff1}
g(\tilde p):=\inf_{\substack{\mc{H}_{\hat{E}}\simeq\mc{H}_{Q_1}\\\ket{\rho}\in\Sigma_{\hat{E}}(\tilde p)}} H(\hat{Z}|S\hat{E})_{\rho^\text{EAT}},
\end{equation}
which can serve as a crossover min-tradeoff function for EAT channels $\{\mc{M}^\text{EAT}_i\}$. In order to obtain an efficiently numerically computable crossover min-tradeoff function, we now make use of the framework presented in \cite{winick2018reliable} to remove the dependency on Eve's subsystem.

\subsubsection{Removing the dependence on the $\hat{E}$ subsystem}\label{subsub:RemoveDependence}

The idea is to consider a coherent version of a round of the protocol leading to Bob's raw key $\hat{Z}$. Namely, Alice and Bob's measurements are performed in a coherent fashion, i.e. by means of isometries acting on the system to be measured and adding a quantum register containing the quantum information which, once dephased, will provide the measurement result, but not yet dephasing it. Alice and Bob then publicly announce partial information about their measurement outcomes, while keeping part of the information stored coherently.  From that information they decide whether they use the round for key generation, parameter estimation or tomography of Alice's part. If they use the round for key generation, in the case of reverse reconciliation, Bob applies a key map to his coherently stored measurement outcomes, which provides a coherent key register. The key can then be obtained by means of a so-called pinching operation, i.e. a measurement that dephases the key register. 

As all steps of the protocol before the pinching are performed coherently, we can express the outcome as a pure state, allowing us to apply Theorem 1 in \cite{coles2012unification}, which removes the dependence on the $\hat{E}$ subsystem. In order to formulate our result, we need to introduce the CP map $\mc{G}:AB\to AB\hat{Z}$ that describes the coherent version of the protocol. This map is given by a single Kraus operator

\begin{equation}\label{def:G}
G=\mathbb{1}_A\otimes\sum_{z=0}^3\sqrt{R^z_B}\otimes\ket{z}_{\hat{Z}},
\end{equation}
where we have defined the region operators
\begin{align}
&R^z_B=\frac{1}{\pi}\int_{0}^{\infty}\int_{\frac{\pi}{4}(2z-1)}^{\frac{\pi}{4}(2z+1)}\gamma\pro{\gamma e^{i\theta}}d\theta d\gamma,\label{eq:R0}
\end{align}
for $z\in\{0,1,2,3\}$. 

Furthermore, we define the pinching operation $\mc{Z}:\hat{Z}\to\hat{Z}$, defined by Kraus operators
\begin{align}\label{def:Z}
&Z_j=\pro{j}_{\hat{Z}} \otimes \mathbb{1},
\end{align}
for $j\in\{0,1,2,3\}$, and the identity is extended to all registers other than $\hat{Z}$.
It then holds

\begin{lemma}\label{Lemma:DimReduct}
The crossover min-tradeoff function defined by eq. (\ref{eq:MinTradeOff1}) can be reformulated as follows

\begin{align}
g(\tilde p)=&\inf_{\substack{\mc{H}_{\hat{E}}\simeq\mc{H}_{Q_1}\\\ket{\rho}\in\Sigma_{\hat{E}}(\tilde p)}} H(\hat{Z}|S\hat{E})_{\rho^\mathrm{EAT}}\nonumber\\
=& p^\mathrm{key}\inf_{\rho\in\Sigma(\tilde p)}D(\mc{G}(\rho_{AB})||\mc{Z}(\mc{G}(\rho_{AB})),\label{eq:MinTradeOff2}
\end{align}
where we have defined the set $\Sigma(\tilde p)=\left\{\rho_{AB}\in \mc{D}(\mc{H}_{AB}): \bra{c}\mc{M}^\mathrm{EAT,test}(\rho)_{C}\ket{c}\equiv \tilde p(c) \right\}$, which is independent of the reference system.

\end{lemma}
 
The proof of Lemma \ref{Lemma:DimReduct}  goes along the line of the discussion in \cite{lin2019asymptotic} and can be found in Appendix \ref{App:A}. Let us note that, by definition of the protocol, it holds $\dim(A)=4$, but the dimension of Bob's states can be infinite. We address this problem by invoking again the bounded-energy assumption, implying that the solution to \eqref{eq:MinTradeOff2} can be arbitrarily well approximated by taking a large enough finite dimension. We then take $d_B$ to be arbitrary but finite. Under said assumption, the set $\Sigma(\tilde p)$ is compact and, as the objective is continuous \cite{winick2018reliable}, a minimum is attained in eq. (\ref{eq:MinTradeOff2}). Following \cite{liu2019device} we can show
\begin{lemma}\label{lemma:conv}
For a given $0< p^\text{key}\leq 1$,
\begin{equation}\label{115}
g(\tilde p)=p^\mathrm{key}\min_{\rho\in\Sigma(\tilde p)}D(\mc{G}(\rho_{AB})||\mc{Z}(\mc{G}(\rho_{AB})))
\end{equation}
is a convex function on $\mc{P}_{\tilde{\mc{C}}}$. 
\end{lemma}
The proof can be found in Appendix \ref{App:conv}.

\subsubsection{Finding an affine crossover min-tradeoff function}
We note that, for any given distribution $\tilde p\in\mathcal{P}_{\tilde{\mc{C}}}$, eq. (\ref{115}) is a convex optimisation problem with semidefinite constraints. As the objective is not affine, however, it is not a semidefinite program (SDP). Also, the dependence of $g$ on the distribution $\tilde p$ is hidden in the constraints. We will now follow the steps taken in \cite{winick2018reliable} and perform a first order Taylor expansion (around some state $\tilde\rho_{AB}\in\mc{D}(\mc{H}_{AB})$) providing a lower bound on the optimisation problem in eq. (\ref{115}). The resulting expression contains an SDP with a linear objective.

We then consider the dual of the SDP, which for any dual feasible point provides an affine lower bound on the original SDP. By the nature of duality, which roughly speaking incorporates the constraints into the objective, the objective of the dual problem will explicitly depend on $\tilde p$ in an affine way, as will the entire expression lower bounding  the optimising problem in eq. (\ref{115}). Thus, for any given state $\tilde\rho_{AB}\in\mc{D}(\mc{H}_{AB})$, as well any dual feasible point, we can obtain an affine crossover min-tradeoff function.

To begin with, let us explicitly consider the optimisation problem in eq. (\ref{115}). Let $m = 4 \Delta/\delta + 4$ represent the total number of modules in Bob's discretisation. For a probability distribution $\tilde  p\in\mc{P}_{\tilde{\mc{C}}}$, the optimisation takes the form 
\begin{align}
\min_{\rho_{AB}}\ &D(\mc{G}(\rho_{AB})||\mc{Z}(\mc{G}(\rho_{AB})))\label{eq:SDP0}\\
\text{s.t.} \; &\rho_{AB}\geq0,\;\Tr[\rho_{AB}]=1,\nonumber\\
&\forall x\in\{0,1,2,3\},\;\forall z\in\{0,...,m-1\}:\nonumber\\
&\;\tilde{p}^\text{PE}\Tr\left[\(\pro{x}_A\otimes \tilde{R}_B^z\) \rho_{AB} \right]=\tilde p(x,z,\perp),\nonumber\\
&\forall x'\in\{0,...,15\}: \nonumber\\
&\;\tilde{p}^\text{tom}\Tr\left[\Gamma_{x'} \rho_{A} \right]=\tilde p(\perp,\perp,x'),\nonumber
\end{align}
where we have defined $\tilde{p}^\text{PE}=\frac{p^\text{PE}}{1-p^\text{key}}$ and $\tilde{p}^\text{tom}=1-\tilde{p}^\text{PE}$. Further, $\tilde{R}^z$ are region operators defined in analogy to \eqref{eq:R0}, but with the discretisation used for parameter estimation given by eq. (\ref{eq:discPE}). Regarding the constraints, the region operators related to parameter estimation add up to the identity matrix, so that there is no need to impose the constraint $\Tr[\rho_{AB}] = 1$. 
We now closely follow \cite{winick2018reliable} to lower bound  eq. (\ref{eq:SDP0}). For brevity, let us define $r(\rho):=D(\mc{G}(\rho)||\mc{Z}(\mc{G}(\rho)))$. By the properties of the pinching quantum channel, this expression can be rewritten without loss of generality in terms of von Neumann entropies
\begin{equation}
    r(\rho) = H(\mc{Z}(\mc{G}(\rho)) -  H(\mc{G}(\rho)).
\end{equation}
Using the methodology of \cite{hu2021robust}, it is possible to apply here a facial reduction to reformulate the maps $\mathcal{Z}$ and $\mathcal{G}$ into maps which are strictly positive definite; this does not only assure that the new objective function is differentiable for any $\rho>0$, but also reduces the dimension of both maps which simplifies the subsequent numerical analysis. This process can be seen as a unitary transformation, such that \footnote{For the numerical implementation, this decomposition can be obtained in MATLAB by using the function \texttt{rank}.}
\begin{equation}
    \mathcal{G}(\rho) = 
    \begin{bmatrix}
    U & V
    \end{bmatrix}
    \begin{bmatrix}
    \tilde{\mathcal{G}}(\rho) & 0 \\
    0 & 0
    \end{bmatrix}
    \begin{bmatrix}
    U^\dagger \\
    V^\dagger
    \end{bmatrix},
\end{equation}
where $\tilde{\mc{G}}(\rho) > 0$ for $\rho>0$. A similar procedure follows for $\mathcal{Z}(\mathcal{G}(\rho))$, resulting in a new map $\tilde{\mathcal{Z}}(\rho) > 0$. Hence, by taking advantage of the fact that the von Neumann entropy is invariant under unitary transformations, we arrive at a simpler objective function
\begin{equation}
r(\rho) = H(\tilde{\mathcal{Z}}(\rho)) - H(\tilde{\mc{G}}(\rho)). \label{newobj}
\end{equation}
With the maps $\tilde{\mc{Z}}$, $\tilde{\mc{G}}$  
the matrix gradient $\nabla r(\rho)$ is now given by 
\begin{align}
\nabla r(\rho)^T &= [\tilde{\mathcal{G}}^\dagger (\log\tilde{\mathcal{G}}(\rho)) + \tilde{\mathcal{G}}^\dagger(\mathbb{1})] \nonumber \\
&\quad - [\tilde{\mathcal{Z}}^\dagger(\log\tilde{\mathcal{Z}}(\rho)) + \tilde{\mathcal{Z}}^\dagger(\mathbb{1})] .\label{eq:MatrixGrad}
\end{align}
Let now $\tilde  p\in\mc{P}_{\tilde{\mc{C}}}$ and $\rho_{\tilde p}^*\in\Sigma(\tilde p)$ be the minimiser of (\ref{eq:SDP0}). For any $\tilde\rho\in\mc{D}(\mc{H}_{AB})$, it then holds
\begin{align}
\frac{g(\tilde p)}{p^\text{key}}&=r(\rho_{\tilde p}^*) \nonumber
\\ &\geq r(\tilde\rho)+\Tr\left[(\rho_{\tilde p}^*-\tilde\rho)^T\nabla r(\tilde\rho)\right]\\
&\geq r(\tilde\rho)-\Tr\left[\tilde\rho^T\nabla r(\tilde\rho)\right]\nonumber \\
&\quad+\min_{\sigma\in\Sigma(\tilde p)}\Tr\left[\sigma^T\nabla r(\tilde\rho)\right], \label{eq:85}
\end{align}
where the first inequality is due to the fact that $r$ is a convex, differentiable function over the convex set $\mc{D}(\mc{H}_{AB})$, hence it can be lower bounded by its first order Taylor expansion at $\tilde\rho$ (see e.g. \cite{boyd2004convex} p.69), and the second inequality is due to the fact that $\rho_{\tilde p}^*\in\Sigma(\tilde p)$. For any $\tilde\rho\in\mc{D}(\mc{H}_{AB})$ and $\tilde p$, the optimisation problem in eq. (\ref{eq:85}) is an SDP in standard form, explicitly given by

\begin{align}
\min_{\sigma_{AB}}\ &\Tr\left[\sigma^T\nabla r(\tilde\rho)\right]\label{eq:SDP2}\\
\text{s.t.} \; &\sigma_{AB}\geq0,\;\nonumber\\
&\forall x\in\{0,1,2,3\},\;\forall z\in\{0,...,m-1\}:\nonumber\\
&\;\tilde{p}^\text{PE}\Tr\left[\(\pro{x}_A\otimes \tilde{R}_B^z\) \sigma_{AB} \right]=\tilde p(x,z,\perp),\nonumber\\
&\forall x'\in\{0,...,15\}: \nonumber\\
&\;\tilde{p}^\text{tom}\Tr\left[\Gamma_{x'} \sigma_{A} \right]=\tilde p(\perp,\perp,x').\nonumber
\end{align}
The dual problem of the SDP (\ref{eq:SDP2}) takes the form
\begin{equation}\label{eq:maximization}
\max_{\vec{\nu}\in\Sigma^*_{\tilde\rho}}\ell_{\tilde p}(\vec{\nu}),
\end{equation}
where the dual objective is given by 

\begin{align}
\ell_{\tilde p}(\vec{\nu})=\sum_{x=0}^3\sum_{z=0}^{m-1}\nu_{xz}\frac{\tilde p(x,z,\perp)}{\tilde{p}^\text{PE}} +  {\sum_{x'=0}^{15}\nu'_{x'}\frac{\tilde p(\perp,\perp,x')}{\tilde{p}^\text{tom}}},
\end{align}
which is affine with respect to $\tilde p$. Further, the set $\Sigma^*_{\tilde\rho}$ is defined as 
\begin{align}
\Sigma^*_{\tilde\rho}=&\left\{\vec{\nu}\in{\mathbb{R}^{4m +16}}:
\right.\nonumber\\
&\quad \nabla r(\rho)-\sum_{x=0}^3\sum_{z=0}^{m-1}\nu_{xz}\left(\pro{x}_A\otimes \tilde{R}_B^z\right)^T \nonumber \\
&\left. \quad -\sum_{x'=0}^{15}\nu'_{x'}\Gamma_{x'}^T\geq0\right\}
\end{align}
which is independent of $\tilde p$. By weak duality it then holds
\begin{align}
g(\tilde p)&=p^\text{key}r(\rho^*_{\tilde p})\nonumber\\
&\geq p^\text{key}\(r(\tilde\rho)-\Tr\left[\tilde\rho^T\nabla r(\tilde\rho)\right]+\max_{\vec{\nu}\in\Sigma^*_{\tilde\rho}}\ell_{\tilde p}(\vec{\nu})\) \label{eq:lowerbound}\\
&\geq p^\text{key}\(r(\tilde\rho)-\Tr\left[\tilde\rho^T\nabla  r(\tilde\rho)\right]+\ell_{\tilde p}(\vec{\nu})\)\\
&=:\tilde{g}_{\vec{\nu},\tilde\rho}(\tilde p)\label{129}
\end{align}
for any $\tilde\rho\in\mc{D}(\mc{H}_{AB})$ and any $\vec{\nu}\in\Sigma^*_{\tilde\rho}$. We note that for any such choice of $\tilde\rho,\vec{\nu}$, the function $\tilde{g}_{\vec{\nu},\tilde\rho}:\mc{P}_{\tilde{\mc{C}}}\to\mathbb{R}$ is an affine crossover min-tradeoff function. 

\subsubsection{Optimisation of the crossover min-tradeoff function}

In this section we describe how we can numerically obtain almost optimal, i.e. optimal up to numerical imprecision, choices for our parameters 
 $\tilde\rho$ and $\vec{\nu}$ in the crossover min-tradeoff function (\ref{129}), for a given distribution $\tilde{p}_0\in\mc{P}_{\tilde{\mc{C}}}$. The distribution will be of the form $\tilde p_0(x,z,\perp)=\tilde{p}^\text{PE}p^\text{sim}_0(x,z)$, for all $x\in\{0,1,2,3\}$ and $z\in\{0,...,m-1\}$, where $p^\text{sim}_0(x,z)$ is a distribution obtained by simulating an honest implementation of the physical QKD protocol. Similarly, $\tilde p_0(\perp,\perp,x')= \tilde{p}^\text{tom} p_0^\text{tom}(x')$ for all $ x'\in\{0,...,15\}$, where $p_0^\text{tom}(x')$ is the distribution obtained in the hypothetical tomography. For the explicit form of $p^\text{sim}_0(x,z)$ and $p^\text{tom}_0(x')$, given by a simulation of the hypothetical QKD protocol, see Section \ref{subs:NumImplementation}.
 
We note that whereas the choices for $\tilde\rho$ and $\vec{\nu}$ will only be optimal up to numerical imprecision, it is possible to analytically confirm their feasibility, i.e. that  $\tilde\rho\in\mc{D}(\mc{H}_{AB})$ and  $\vec{\nu}\in\Sigma^*_{\tilde\rho}$. Thus we can analytically verify that the corresponding function $g_{\vec{\nu},\tilde\rho}$ is indeed a valid crossover min-tradeoff function. 
 
Our numerical method now works as follows: We begin with some $\tilde\rho^{(0)}\in\mc{D}(\mc{H}_{AB})$ and, for $i=1,...,n^\mathrm{iter}$, where $n^\mathrm{iter}\in\mathbb{N}$, iteratively compute

\begin{align}
\Delta\tilde\rho^{(i)} = \text{arg} &\min_{\sigma_{AB}}\ \Tr\left[\sigma^T\nabla r(\tilde\rho^{(i-1)}) \right]\label{eq:SDP2a}\\
\text{s.t.} \; &\sigma_{AB}\geq0,\;\nonumber\\
&\forall x\in\{0,1,2,3\},\;\forall z\in\{0,...,m-1\}:\nonumber\\
&\;\Tr\left[\(\pro{x}_A\otimes \tilde{R}_B^z\) \sigma_{AB} \right]=p_0^\text{sim}(x,z),\nonumber\\
&\forall x'\in\{0,...,15\}: \nonumber\\
&\;\Tr\left[\Gamma_{x'} \sigma_{A} \right]=p_0^\text{tom}(x'). \nonumber
\end{align}
Once this SDP is solved and $\Delta\tilde\rho^{(i)}$ is known, the value of the relative entropy is minimized according to
\begin{align}
    \min_{\kappa \in (0,1)}\ &r (\tilde\rho^{(i-1)} + \kappa \Delta \tilde\rho^{(i)}).\label{eq:globalsearch}
\end{align}
Such minimization can be computed in MATLAB with the function \texttt{fminbnd}. Then, we set a new density matrix $\tilde\rho^{(i)} = \tilde\rho^{(i-1)} + \kappa^* \Delta\tilde\rho^{(i)}$, with the optimal coefficient $\kappa^*$, and repeat the optimisation \eqref{eq:SDP2a}. After $n^\mathrm{iter}$ we set $\tilde\rho_0=\tilde\rho^{(n^\mathrm{iter})}$.

The numerical computation of the dual of \eqref{eq:SDP2} requires to take into account the difference in the numerical representation of the states and operators with respect to their analytical values, which leads to a violation of the constraints due to the computational limitations of the computers. According to Theorem 3 of \cite{winick2018reliable}, this error may be taken into account by introducing a new parameter $\varepsilon'$, which takes the absolute value of the maximal such error, and expands the feasible set to provide a lower bound while preserving the reliability of the approach. 
With this methodology, the dual takes the form~\cite{boyd2004convex},
\begin{equation}
\max_{(\vec{\nu},\vec{\mu})\in\tilde{\Sigma}^*_{\tilde\rho_{0}}}\ell^0_{\tilde p_0,\varepsilon'}(\vec{\nu}, \vec{\mu}),
\label{eq:SDP3a}
\end{equation}
where the dual objective is given by 
\begin{align}
\ell^0_{\tilde p_0,\varepsilon'}(\vec{\nu}, \vec{\mu})&=\sum_{x=0}^3\sum_{z=0}^{m-1}\nu_{xz}p_0^\text{sim}(x,z)\nonumber\\
&\quad +\sum_{x'=0}^{15}\nu'_{x'} p_0^\text{tom}(x')\nonumber\\
&
\quad -\varepsilon' \sum_{z'=1}^{4m + 16}\mu_{z'},
\end{align}
with the set $\tilde{\Sigma}^*_{\tilde\rho_0}$ defined as
\begin{align}
\tilde{\Sigma}^*_{\tilde\rho_0}=&\left\{(\vec{\nu},\vec{\mu})\in(\mathbb{R}^{4m +16},\mathbb{R}^{4m +16}):-\vec{\mu}\leq \vec{\nu}\leq\vec{\mu},\right.\nonumber\\
&\quad \nabla r(\tilde\rho_0)
-\sum_{x=0}^3\sum_{z=0}^{m-1}\nu_{xz}\left(\pro{x}_A\otimes \tilde{R}_B^z\right)^T \nonumber \\
&\quad \left. -\sum_{x'=0}^{15}\nu'_{x'}\Gamma_{x'}^T\geq0\right\}.
\end{align}
From this maximization, as well as a fixed value $\varepsilon'$ taken according to the maximal numerical error at the constraints, we obtain $\vec{\nu}_0$, and note that $\vec{\nu}_0\in\Sigma^*_{\tilde\rho_0}$. This allows us to define our crossover min-tradeoff function as 
\begin{align}
&\tilde{g}_0(\tilde p):=\tilde{g}_{\vec{\nu}_0,\tilde\rho_0}(\tilde p)\nonumber\\
&={p^\text{key}}\(r(\tilde\rho_0)-\Tr\left[\tilde\rho_0^T\nabla r(\tilde\rho_0)\right]+\ell_{\tilde p}(\vec{\nu}_0)\)\label{minTradeoff}\\
&={p^\text{key}}\(G_0+\sum_{x=0}^3\sum_{z=0}^{m-1}\nu_{0,xz}\frac{\tilde p(x,z,\perp)}{\tilde{p}^\text{PE}}\right.\nonumber\\
&\quad \left.+{\sum_{x'=0}^{15}\nu'_{0,x'}\frac{\tilde p(\perp,\perp,x')}{\tilde{p}^\text{tom}}}\)\label{def:g},
\end{align}
with a constant
\begin{equation}
G_0:=r(\tilde\rho_0)-\Tr\left[\tilde\rho_0^T\nabla r(\tilde\rho_0)\right].
\end{equation}

In order to compute the higher order terms of the EAT, we need to find $\max(\tilde{g}_{0})=\max_{\tilde p\in\mc{P}_{\tilde{\mc{C}}}}\tilde{g}_{0}(\tilde p)$ and $\min(\tilde{g}_{0})=\min_{\tilde p\in \mc{P}_{\tilde{\mc{C}}}}\tilde{g}_{0}(\tilde p)$. We note that, as $\mc{P}_{\tilde{\mc{C}}}$ is convex and $\tilde{g}_{0}$ is affine, we can restrict to the extreme points of $\mc{P}_{\tilde{\mc{C}}}$. Namely we get
\begin{align}
&\max(\tilde{g}_{0})=p^\text{key}G_0+p^\text{key}\max(\nu_0),\\
&\min(\tilde{g}_{0})=p^\text{key}G_0+p^\text{key}\min(\nu_0), \label{Ming}
\end{align}
where we have defined
\begin{align} 
\max(\nu_0):=\max&\(\left\{\frac{\nu_{0,xz}}{\tilde{p}^\text{PE}}\right\}_{(x,z)=(0,0)}^{(3,m-1)}\cup\left\{\frac{\nu'_{0,x'}}{\tilde{p}^\text{tom}}\right\}_{x'=0}^{15}\)\label{eq:max_nu}\\
\min(\nu_0):=\min&\(\left\{\frac{\nu_{0,xz}}{\tilde{p}^\text{PE}}\right\}_{(x,z)=(0,0)}^{(3,m-1)}\cup\left\{\frac{\nu'_{0,x'}}{\tilde{p}^\text{tom}}\right\}_{x'=0}^{15}\)\label{eq:min_nu}
\end{align}
In the case where the minimisers are non-positive and and the maximisers are non-negative, we can upper bound 
\begin{align}\label{eq:128}
    \max(\tilde{g}_{0})-\min(\tilde{g}_{0})
    &\leq \max(\nu_0) - \min(\nu_0),
    \end{align}
which is independent of $p^\text{key}$. Finally, we can introduce the min-tradeoff function induced by our crossover min-tradeoff function $\tilde{g}_0$, given by eq. (\ref{def:g}), via eqs. (\ref{def:f1},\ref{def:f2}).
\begin{align}
f(p)&=\sum_{c\in\tilde{\mathcal{C}}}p(c)\(\max(\tilde{g}_0)+\frac{\tilde{g}_0(\delta_c)-\max(\tilde{g}_0)}{1-p^\text{key}}\)\nonumber\\
&\quad +p(\perp,\perp,\perp)\max(\tilde{g}_0)\\
&=\max(\tilde{g}_0)+\sum_{c\in\tilde{\mathcal{C}}}\frac{p(c)\(\tilde{g}_0(\delta_c)-\max(\tilde{g}_0)\)}{1-p^\text{key}}\\
&=p^\mathrm{key}\(G_0+\max(\nu_0)\)\nonumber\\
&+\sum_{x=0}^3\sum_{z=0}^{m-1}p^\text{key}\frac{\nu_{0,xz}/\tilde{p}^\text{PE}-\max(\nu_0)}{1-p^\text{key}}p(x,z,\perp)\nonumber\\
&+\sum_{x'=0}^{15}p^\text{key}\frac{\nu'_{0,x'}/\tilde{p}^\text{tom}-\max(\nu_0)}{1-p^\text{key}}p(\perp,\perp,x'). \label{eq:mtoff_coeffs}
\end{align}


Let us now observe that we can split the min-tradeoff function according to a constant term (since it does not depend on the probabilities) and the previously-defined functions (\ref{eq:fPE}-\ref{eq:ftom}) for parameter estimation and tomography
\begin{align}
    f(p) = \mathrm{const}+f^\mathrm{PE}(p)+f^\mathrm{tom}(p),\label{eq:ourf}
\end{align}
with
\begin{align}
&f^\mathrm{PE}(p)=\sum_{x=0}^3\sum_{z=0}^{m-1}h_{x,z,\perp}p(x,z,\perp),\label{eq:ourfPE}\\
&f^\mathrm{tom}(p)=\sum_{x'=0}^{15}h_{\perp,\perp, x'}p(\perp,\perp,x').\label{eq:ourftom}
\end{align}
The affine coefficients of the functions, given in terms of the crossover min-tradeoff function \eqref{eq:mtoff_coeffs}, provide the means for calculating the statistical deviations $\delta^\mathrm{tol}_\mathrm{PE}$ and $\delta^\mathrm{tol}_\mathrm{tom}$. Thanks to the affine structure of \ref{eq:ourf}, we can use $f(p)$ as a min-tradeoff function in Theorem \ref{FinteSizePhys}. When applying such Theorem, we evaluate the min-tradeoff function in our simulated honest distribution, $f(p_0)$. However, we use the properties of the distribution $p_0$ to reformulate our function in more convenient shape

\begin{align}
f(p_0)&=p^\mathrm{key} G_0 \nonumber \\
&\quad + \sum_{x=0}^3\sum_{z=0}^{m-1}\nu_{0,xz} p_0^\mathrm{sim}(x,z) \nonumber\\
&\quad +\sum_{x'=0}^{15}\nu'_{0,x'}p_0^\mathrm{tom}(x')\label{146}\\
&=\tilde{g}_0(\tilde{p}_0).
\end{align}


\subsection{Asymptotic Rates}
With our choice of a min-tradeoff function $f(p)$, we can now compute the asymptotic key rate in Theorem \ref{FinteSizePhys}, and show that we can achieve soundness and completeness in the asymptotic limit.

Let $n\in\mathbb{N}$. We begin by noting that, for fixed $m\in\mathbb{N}$, our numerically obtained values $\nu_{0,xz}$, for $x=0,...,3,z=0,...,m-1$, and $\nu'_{0,x'}$, for $x'=0,...,15$, are constant, i.e. independent of $n$. Let us consider some fixed values for the parameters $\epsilon^\mathrm{phys}_{\mathrm{NA}},\epsilon^\mathrm{tom},\epsilon_\text{EC},\epsilon_\mathrm{EC}^c,\epsilon_\mathrm{PE}^c\in(0,1)$, such that $\epsilon^\mathrm{tom}<\frac{1}{2}\epsilon^\mathrm{phys}_{\mathrm{NA}}$, as well as $\epsilon\in\(0,1-\sqrt{2\epsilon^\mathrm{tom}/\epsilon^\mathrm{phys}_{\mathrm{NA}}}\)$,  $\epsilon^\mathrm{phys}=\epsilon+\sqrt{2\epsilon^\mathrm{tom}/\epsilon^\mathrm{phys}_{\mathrm{NA}}}$. Let us also consider some constant $0\leq\tilde{p}^\mathrm{PE}\leq1$, and $\tilde{p}^\mathrm{tom}=1-\tilde{p}^\mathrm{PE}$. 

Since $n\rightarrow \infty$, we can select any scaling such that all rounds tend asymptotically to be spent on key generation, and the finite size effects are reduced. For simplicity, let us then take $a=1+n^{-3/4}$, as well as $p^\text{key}=1-n^{-\frac{1}{2}}$ and $p^\text{PE}=\tilde{p}^\text{PE}n^{-\frac{1}{2}}$, implying $p^\text{tom}=\tilde{p}^\text{tom}n^{-\frac{1}{2}}$. It then holds that $h_{x,z,\perp}=\mathcal{O}(n^\frac{1}{2})$ and $h_{\perp,\perp,x'}=\mathcal{O}(n^\frac{1}{2})$, as well as $p_0(x,z,\perp)=\mathcal{O}(n^{-\frac{1}{2}})$ and $p_0(\perp,\perp,x')=\mathcal{O}(n^{-\frac{1}{2}})$, for all  $x\in\{0,...,3\}$, $z\in\{0,...,m-1\}$ and $x'\in\{0,...,15\}$. Hence, for all $i=1,...,4m$, the quantities defined in eqs. (\ref{eq:79a}-\ref{eq:79b}) scale as follows: $\gamma_i=\mathcal{O}(n^{-\frac{1}{2}})$, $c_i=\mathcal{O}(n^\frac{1}{2})$, and $D=\mathcal{O}(n^\frac{1}{2})$. Consequently, in order to fulfill eq. (\ref{eq:deltaPE}), we have to choose $\delta^\mathrm{tol}_\mathrm{PE}=\mathcal{O}((\log n)^\frac{1}{2}n^{-\frac{1}{4}})$. Similarly, it can be shown that that we need $\delta^\mathrm{tol}_\mathrm{tom}=\mathcal{O}((\log n)^\frac{1}{2}n^{-\frac{1}{4}})$, in order to satisfy eq. \ref{eq:deltatom2}).

As for the remaining higher order terms in eq. (\ref{eq:FiniteKeyRate}), we note that by eqs. (\ref{eq:96},\ref{eq:97},\ref{eq:128}) it holds
\begin{align}
    V&\leq\tilde{V} = \mathcal{O}\(n^\frac{1}{4}\)  \\
    K_a&\leq\tilde{K}_a=\mathcal{O}(1).
\end{align}

Although $\tilde{V}$ and $\tilde{K}_a$ do not decrease with the number of rounds, we note that in \eqref{eq:FiniteKeyRate} they appear multiplied by $a-1$. This eventually leads to
\begin{align}
    (a-1)\tilde{V} & = \mathcal{O}\(n^{-\frac{1}{2}}\),\\
    (a-1)^2 \tilde{K}_a & =\mathcal{O}\(n^{-\frac{3}{2}}\).
\end{align}
Hence, all remaining higher order terms, except $\frac{1}{n}\text{leak}_\mathrm{EC}$, which we keep open, scale as $\mathcal{O}(n^{-\frac{1}{4}})$ or less. Further, the term $f(p_0)$ in Theorem \ref{FinteSizePhys}, given by eq. (\ref{146}), only depends on $n$ via the prefactor $p^\mathrm{key}$.
In summary, we can obtain the following bound on the asymptotic key rate.
\begin{theorem}{\bf (Asymptotic rate)}\label{Asymptotoc}
 For the above mentioned values of the parameters, it holds
 \begin{align}
r^\mathrm{phys}\big|_{\Omega^\mathrm{phys}_\mathrm{NA}}&\geq G_0+\sum_{x=0}^3\sum_{z=0}^{m-1}\nu_{0,xz}p_0^\mathrm{sim}(x,z)\nonumber\\
&\quad +\sum_{x'=0}^{15}\nu'_{0,x'}p_0^\mathrm{tom}(x') -\frac{1}{n}\mathrm{leak}_\mathrm{EC}\nonumber\\
&\quad +\mathcal{O}((\log n)^\frac{1}{2}n^{-\frac{1}{4}})
\end{align}
\begin{align}
\lim_{n\to \infty}r^\mathrm{phys}\big|_{\Omega^\mathrm{phys}_\mathrm{NA}}&\geq G_0+\sum_{x=0}^3\sum_{z=0}^{m-1}\nu_{0,xz}p_0^\mathrm{sim}(x,z)\nonumber\\
&\quad +\sum_{x'=0}^{15}\nu'_{0,x'}p_0^\mathrm{tom}(x') \nonumber \\
&\quad -\lim_{n\to \infty}\frac{1}{n}\mathrm{leak}_\mathrm{EC}\label{rate}.
\end{align}
\end{theorem}

\section{Numerical implementation and results} \label{subs:NumImplementation}
In order to show that our approach produces non-trivial key rates in a realistic implementation, we consider the same scenario that was used in \cite{lin2019asymptotic}. Namely, we simulate an experiment in which Alice and Bob are linked by an optical fibre of length $D$ with excess noise $\xi$, transmittance $\eta = 10^{-\omega D/10}$ and an attenuation of $\omega=0.2$ dB/km. This provides us with a simulated distribution that can be computed efficiently using MATLAB
\begin{equation}\label{eq:doubleint}
    p_0^\text{sim}(x,z) = \int_{\tilde{\mathcal{R}}_z} \frac{\gamma \exp\left(\frac{-|\gamma e^{i \theta} - \sqrt{\eta}\varphi_x|^2}{1 + \eta \xi/2}\right)}{4 \pi (1 + \eta \xi/2)}d\theta d\gamma  ,
\end{equation}
where $\tilde{\mathcal{R}}_z$ represents the fragment of the phase space corresponding to each module $z \in \{0,...,m-1\}$, defined according to the intervals described in \eqref{eq:discPE}, and $\varphi_x \in \{\alpha, i\alpha, -\alpha ,-i\alpha\}$ are the coherent state amplitudes used by Alice with $\alpha \in \mathbb{R}$. The region operators for the constraints $\tilde{R}^z_B$ are given by the same intervals as in \eqref{eq:doubleint}
 \begin{equation}
      \tilde{R}^z_B=\frac{1}{\pi}\int_{\tilde{\mathcal{R}}_z}\gamma\pro{\gamma e^{i\theta}}d\theta d\gamma,
 \end{equation}
 while their numerical implementation requires to switch to the Fock basis. This is done via the inner product \cite{burnett_1998},
 \begin{equation}
     \left<\gamma e^{i \theta}\right.|\left. k\right> = \frac{\gamma^k e^{-\gamma^2 / 2}  e^{-i k \theta}}{\sqrt{k!}}.
 \end{equation}
For the hypothetical tomography, we choose an IC POVM $\{\Gamma_{x'}\}_{x'=0}^{15}$, which completely describes Alice's marginal with a probability distribution
\begin{align}
    p_0^\text{tom}(x') &= \frac{1}{4}\sum_{x,y=0}^3\<\varphi_y|\varphi_x\>\Tr\left[\Gamma_{x'}\ket{x}\bra{y}_A\right].
\end{align}


As argued, under the bounded-energy assumption, the computation of the trade-off requires solving the optimisation for arbitrary finite dimension $d_B$. At the moment we are unable to do this, so we truncate operators by introducing a photon number cutoff $N_c$. That is, we impose that all operators, when expressed in the Fock basis, involve terms having at most $N_c$ photons, which implies that $d_B=N_c+1$. We solve the optimisation for increasing values of $N_c$ and we always observe that the obtained trade-offs numerically converge, see for instance Fig~\ref{fig:Nc_Test} below. A value of  $N_c = 15$ provides a good balance between the execution time of the solver and the reliability of the numerics, which is consistent with what was previously observed in~\cite{lin2019asymptotic,hu2021robust,ghorai2019asymptotic}. Based on all the obtained numerical evidence, we make the following

\textbf{Numerical convergence assumption:} The derived numerical trade-offs for the considered cut-offs provide reliable approximations to the trade-off for arbitrary finite dimension. 

In our view, this assumption is quite plausible in the considered setup, as the amplitude of the states detected by Bob decreases with the channel losses, which eventually means a decreasing average number of received photons.


With all the elements of the optimisation defined, we minimize the SDP \eqref{eq:SDP2a} according to the Frank-Wolfe algorithm. For this process we use the toolbox YALMIP \cite{YALMIP} together with the interior point solver SDPT3 \cite{sdpt3_one,sdp3_two}. Once the suboptimal bound is obtained, we compute the dual \eqref{eq:SDP3a} using the optimization software CVX \cite{cvx,gb08}, since it provides slightly better results than YALMIP and SDPT3. For the iterative process of the Frank-Wolfe algorithm, we set a stopping criterion based on calculating the lower bound \eqref{eq:lowerbound} every 15 iterations of the minimization \eqref{eq:SDP2a}; if the relative difference between the upper bound given by minimising \eqref{newobj} and the reliable lower bound is smaller than a $2\%$, the algorithm stops the optimization. If this margin is not reached, the algorithm continues until a total of $300$ iterations are performed. Using this approach we obtain $\tilde\rho_0$ and $\vec{\nu}_0$, the feasibility of which can be checked analytically. This allows us to obtain a crossover min-tradeoff function $\tilde{g}_0$ via eq. (\ref{minTradeoff}). By Theorem \ref{Asymptotoc},  we observe the asymptotic rate

\begin{align}
    r_\infty &\geq  G_0+\sum_{x=0}^3\sum_{z=0}^{m-1}\nu_{0,xz}p_0^\mathrm{sim}(x,z)\nonumber\\
    &\quad +\sum_{x'=0}^{15}\nu'_{0,x'}p_0^\mathrm{tom}(x')-\lim_{n\to \infty}\frac{1}{n}\mathrm{leak}_\mathrm{EC}.\label{eq:asymptoticLimit}
\end{align}


For the classical information leaked during error correction, we can assume an honest, iid implementation of the protocol. We introduce a parameter $f$ that quantifies the the error correction efficiency with respect to the ideal Shannon limit, so that we write 
\begin{equation}
       \frac{1}{n}\mathrm{leak}_\mathrm{EC} \leq p^\mathrm{key} (1+f)  H(\hat{Z}|\hat{X})
\end{equation}
where $\hat{X}$, $\hat{Z}$ represent the key string bits (after removing the symbol $\perp$) of Alice and Bob respectively. On the other hand, the parameter $p^\mathrm{key}$ comes from the fact that only the signals coming from key rounds require error correction. The error correction efficiency may depend on several factors, such as the chosen code, the block size or the form of the probability distribution between Alice and Bob. Here we take $f$ from values that range from $0\%$ to $5\%$ as a showcase of the potential results for our scheme given diverse efficiencies for error correction.
The Shannon term $H(\hat{Z}|\hat{X})$ can be computed numerically according to the distribution \eqref{eq:doubleint} adapted for the modulation of the key rounds, namely
\begin{align}
     &p^\text{EC}_0 (x,z) \nonumber \\ 
     &= \int_0^\infty \int_{\frac{\pi}{4}(2z-1)}^{\frac{\pi}{4}(2z+1)} \frac{\gamma \exp\left(\frac{-|\gamma e^{i \theta} - \sqrt{\eta}\varphi_x|^2}{1 + \eta \xi/2}\right)}{4 \pi (1 + \eta \xi/2)}d\theta d\gamma.
\end{align}

With the error correction cost, it is not only possible to calculate the asymptotic secret key rate, but also optimise the amplitude $\alpha$ that Alice chooses for her coherent states, which is not attainable with only the results from the SDP.



To test the accuracy of our approach, we compared our results for different values of the cutoff, as well as with the so far standard method of computing the asymptotic key rate based on performing parameter estimation with  moments of the quadrature operators \cite{lin2019asymptotic}. The comparison can be found in Figure \ref{fig:Nc_Test} where one can see that, while using moments to constrain the state shared by Alice and Bob produces better rates for distances shorter than 15 km, both approaches provide comparable results for larger distances. Note that computing moments and coarse-grained probabilities are different ways of discretising the information contained in a CV distribution. These results show that taking moments is better for short distances, albeit the two approaches lead to almost the same values when losses become large. Moreover, as announced, in both cases the asymptotic key rates seem to saturate when increasing the cutoff value. We verified such hypothesis at the inset of Figure \ref{fig:Nc_Test}, where it can be observed that the curves for our modulation converge to the same values.


\begin{figure}[h]
    \centerline{\input{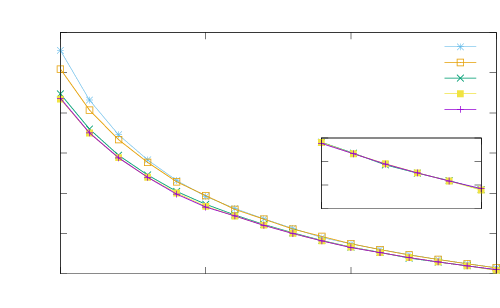}}
    \caption{Asymptotic secret key generation rate with $\xi=2\%$ according to the parameter estimation described in Lin et al. \cite{lin2019asymptotic} and our modulation $(\Delta, \delta)=(0.9,0.9)$, both with ideal error correction (i.e., Shannon limit) $f=0 \%$, and diverse values for the cutoff $N_c$. The amplitude, taken to be the same for all curves, was optimised with respect to the distance. The inset shows the convergence of our modulation.}
     \label{fig:Nc_Test}
\end{figure}

Figure \ref{fig:AsymptoticRate} shows the asymptotic key rates according to eq. (\ref{eq:asymptoticLimit}), where a modulation $(\Delta, \delta)=(0.9,0.9)$ was employed together with a cutoff $N_c = 15$ for ideal error correction $f=0\%$. For distances below $150$ km, the algorithm typically needs 120 or less iterations to converge. For larger distances, it was necessary to reach the limit of 300 iterations before using eq. \eqref{eq:maximization} to obtain the reliable lower bound. On the other hand, we observed a numerical error at the constraints $\varepsilon'$ typically between $10^{-10}$ for small lengths and $10^{-15}$ for very long distances, which ensures both the reliability of the code and the tightness of the key rates.

\begin{figure}[h]
\centerline{\input{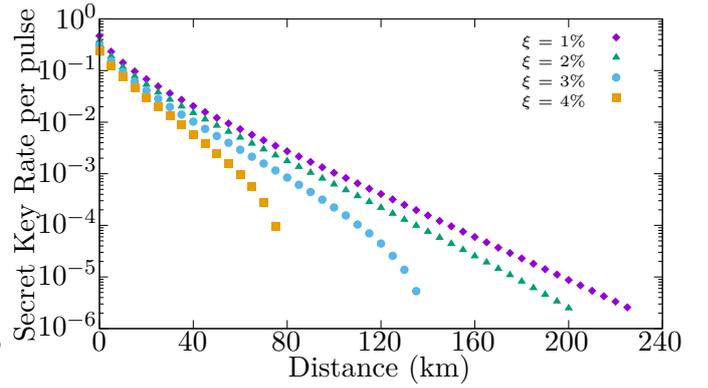}}
    \caption{Asymptotic secret key generation rate according to \eqref{eq:asymptoticLimit} in terms of distances $D$ and excess noise $\xi$ with a modulation $(\Delta, \delta)=(0.9,0.9)$, ideal error correction $f=0 \%$ and cutoff $N_c=15$. The amplitude of the coherent states was optimized with respect to the distance.}
    \label{fig:AsymptoticRate}
\end{figure}

Switching to the finite-size regime, we can use our crossover min-tradeoff function $\tilde{g}_0$ in Theorem \ref{FinteSizePhys} to observe the finite key generation rates. Choosing the parameters $\xi=1\%$, $f=1\%$, $N_c = 12$, $\epsilon = 10^{-10}$, $\epsilon^\mathrm{phys}_{\mathrm{NA}} = 10^{-4}$, $\epsilon^\mathrm{tom} = 10^{-10}$ and $\epsilon_\mathrm{PE}^c = 10^{-10}$ together with a grid search optimisation over $a$ and $p^\text{key}$, we obtain non-zero key rates for $n\geq 10^{12}$ rounds and distances $D\geq15$ km. The outcomes of this process are illustrated in Figure \ref{fig:FiniteRate}, where the curves represent different values for the finite key generation rates with respect to the number of rounds taken for the protocol.

In this regard, the simplest approach to perform the finite-key analysis is to use the resulting data from the asymptotic regime, particularly the dual point \eqref{eq:SDP3a}, in order to build the min-tradeoff function. However, this leads in general to suboptimal results for the finite case since the calculated dual point is optimal only in the asymptotic regime---the dual variables appear in the correction terms of the finite-key rate, whose optimization is not included in the Frank-Wolfe method. For instance, $\tilde{K}_a$ scales exponentially with the spread of the min-tradeoff function, which depends on the dual variables according to (\ref{eq:max_nu}-\ref{eq:min_nu}). Therefore, the dual variables severely affect the finite-key rates. In order to ameliorate this inconvenience, we make use of a perturbative analysis based on genetic algorithms, which reduces the value of the dual variables while preserving a reasonable performance for any block sizes $n$. We defer the details of the method to Appendix \ref{App:D}, and refer to the complete code available in \cite{CPG_2024}. 

We also note that the overall numerical performance of our code enables us to derive the asymptotic secret key in the order of minutes with a reasonable value of the cutoff, $N_c = 12$. Although the perturbative analysis described here increases the overhead of the computations, a fine-tuned implementation can increase the efficiency of the code and perform the complete finite-size analysis for a given distance in a few minutes---such that it can be used in real, on-demand applications.


\begin{figure}[h]
    \centerline{\input{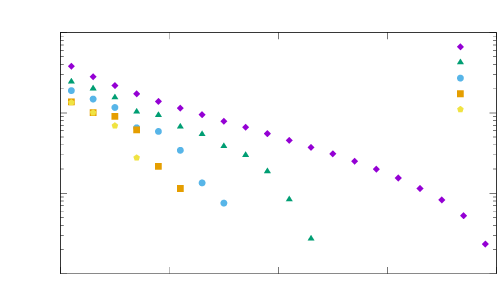}}
    \caption{Finite-size secret key generation rate according to \eqref{eq:FiniteKeyRate} for $\xi=1\%$, $f=1\%$, $N_c=12$, $n$ rounds for the protocol and a modulation $(\Delta, \delta)=(0.9,0.9)$. The parameters $a$, $p^\mathrm{key}$ and $\tilde{p}^\mathrm{PE}$ were optimised according to a grid search, and we set $\epsilon = 10^{-10}$, $\epsilon^\mathrm{phys}_{\mathrm{NA}} = 10^{-4}$, $\epsilon^\mathrm{tom} = 10^{-10}$ and $\epsilon_\mathrm{PE}^c = 10^{-10}$.}
    \label{fig:FiniteRate}
\end{figure}

Finally, we explore the impact of error correction efficiency in the observed key rates. It is well known that in standard CVQKD protocols, as considered in this work, the value of Alice's and Eve's conditional entropies on Bob's results, $H(\hat{Z}|\hat{X})$ and $H(\hat{Z}|\hat{E})$, are very close, especially for large distances. Hence, a non-zero value of $f$ severely affects the possibility of having non-zero key rates. To study this, we plot the finite key rates for blocks of size $n=5 \times 10^{12}$ as a function of the error correction efficiency in Fig.~\ref{fig:n12_f}. As it can be seen, small values of $f$, or in other words, error correction codes with efficiency very close to the Shannon limit, are necessary to generate a secret key for distances beyond 20 km.

\begin{figure}[h]
    \centerline{\input{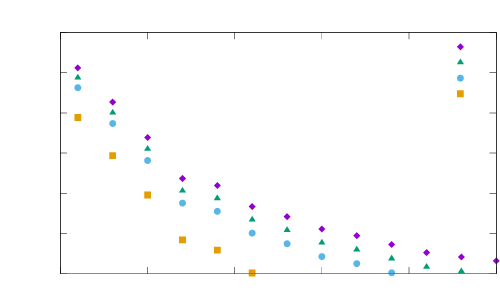}}
    \caption{Finite-size secret key generation rates for $n=5 \times 10^{12}$ rounds and different values of the error correction efficiency $f$. The parameters were taken to be the same as in Figure \ref{fig:FiniteRate}.}
    \label{fig:n12_f}
\end{figure}

\section{Discussion}
In this work we have provided a security proof against arbitrary general attacks for a discrete modulated CVQKD protocol, in which Alice prepares four coherent states and Bob performs heterodyne measurements. The proof exploits the fact that the information used for parameter estimation consists of coarse grained probabilities of the generated continuous measurement outcomes instead of moments, as is the case in most of previous approaches of CVQKD. As shown, this hardly affects the asymptotic key rates, but significantly simplifies the security analysis, as one can employ methods originally introduced in the context of DVQKD, such as the EAT.

Despite the simplifications, the application of the EAT, in its original form~\cite{dupuis2016entropy,dupuis2019entropy}, to the considered prepare-and-measure QKD protocol is not straightforward. The challenging aspect of this has been the fact that, in order to describe the QKD protocol as a sequence of EAT channels, where Eve's reference system cannot be updated by the EAT channels, we had to use an entanglement-based version of the protocol when applying the EAT. Whereas any prepare-and-measure QKD protocol can be easily transformed to an entanglement-based protocol using the source replacement scheme, we have been faced with the issue that the minimisation defining our min-tradeoff function in eq (\ref{def:MinTrade}) can be constrained only in terms of the observed statistics from Alice and Bob's measurements in parameter estimation rounds, i.e. by a distribution of the classical output $C_1^n$. Such constraints are sufficient to obtain a nonzero key rate in DIQKD protocols, as has been considered in \cite{arnon2018practical,arnon2019simple,tan2020improved}. This is due to the fact that in device independent settings the observed statistics alone has to be sufficient to certify an entangled state between Alice and Bob, which is a prerequisite to obtain secure key. In device dependent settings, such as the one we have considered here, however, the observed statistics from Alice and Bob's parameter estimation rounds does not necessarily suffice to certify entanglement. Consequently, if we only use statistics from parameter estimation rounds, the bound on the key rate becomes trivial. We have overcome this issue by considering a hypothetical protocol in which Alice uses some randomly chosen rounds to perform a state tomography on her marginal state, the outcome of which is included in $C_1^n$. Thus, the observed statistics becomes sufficient to obtain nontrivial bounds on the key rate. 

The introduction of the hypothetical tomography poses some additional challenges in the finite size security proof. In particular there is a possibility  that the tomography test does not pass, in which case the hypothetical protocol would abort. In order to ensure that this only happens with negligible probability, we have introduced a tolerance parameter $\delta^\mathrm{tol}_\mathrm{tom}$ in Lemma \ref{lem:deltatolbound}, which has to be subtracted from our key rate. Also, in order to prove security of the physical protocol, it is necessary to show that the raw key states obtained in the hypothetical and physical protocol do not differ by too much and adapt the smoothing parameter accordingly. We have done this in Lemma \ref{lem:HypPhys}, again at the cost of a reduction of the key rate. We note that other approaches, such as using the Asymptotic Equipartition Property \cite{leverrier2015composable,lupo2022quantum,kanitschar2023finite}, provide higher finite key rates without the need to add a virtual tomography---such frameworks prove to be simpler, and the numerical key rates (via e.g. \cite[Theorem 6]{kanitschar2023finite}) can be calculated faster, but they are limited to the case of collective attacks, which is surpassed in this work.

As mentioned in the introduction, after much of the work going into this result was finished, a generalised version of the EAT has been presented in~\cite{metger2022generalised}, known as Generalised EAT (GEAT). In contrast to the original EAT, the new version allows for Eve's reference system to be updated, while also relaxing the Markov condition to a non-signalling condition. Using GEAT, it is possible to express a prepare-and-measure QKD protocol directly into a sequence of EAT channels, without the need to use an entanglement based version of the protocol, as was shown in \cite{metger2022security}. When using this new method, there is no need to introduce a hypothetical tomography, hence our Lemmas \ref{lem:deltatolbound} and \ref{lem:HypPhys} would not be needed and higher key rates may be expected. The only caveat when using GEAT is that it makes an additional assumption on Eve's attack, namely only allowing her to have one quantum system at a time. This condition can be enforced by Alice waiting for Bob to confirm he has received a state before sending the next one \cite{metger2022security}, which might not always be practical. Our method of applying the original EAT does not need this assumption. We therefore believe that our current method of using an hypothetical protocol with tomography of Alice's marginal, is of interest not only for discrete modulated CVQKD, but for proving security of device dependent QKD in settings where the condition that Eve only holds one system at a time is not practical.

Our security analysis can be improved in several directions. As mentioned, being a prepare-and-measure protocol, it is natural to consider the application of GEAT. This may not only provide larger finite-key generation rates, but also allow one to study variants of the protocol using homodyne measurements, which we were unable to accommodate within our security analysis. Another related question is to analyse using GEAT how the obtained rates vary with the number of states prepared by Alice and, in particular, how they approximate the rates of Gaussian modulated protocols. 

The derived key rates are valid under the bounded-energy assumption, stating that Eve's attack involves states of bounded energy, and a numerical-convergence assumption, stating that the numerical curves obtained for increasing number of photons are very close to the trade-off for arbitrary finite dimension. The first assumption is physically realistic and implies that the states in the protocol can be arbitrarily well approximated by states in a finite dimensional Hilbert space of large enough dimension \cite{tomamichel2017largely}. This allows the use of EAT since, whereas this theorem does not require an explicit bound on the Hilbert space dimension of Eve, all Hilbert space dimensions are assumed finite \cite{dupuis2016entropy}. It would be interesting to remove this assumption using recent advances towards a generalisation of the EAT to infinite dimensional Hilbert spaces \cite{fawzi2022asymptotic}.

The second assumption seems quite plausible in the considered setup, as Alice first prepares coherent states with a small average number of photons that are later sent through a lossy channel. Yet, it is interesting to study how to remove the cutoff in the computation of the asymptotic key rates. This has been achieved for collective attacks in the case where the information used in parameter estimation is made of moments of Bob's quadratures \cite{Upadhyaya2021,kanitschar2023finite}. The idea in \cite{Upadhyaya2021} is to introduce a cutoff parameter that depends on the expectation values obtained in parameter estimation and replace the infinite dimensional optimization by a finite dimensional one, plus a correction term, both of which depend on the cutoff parameter. Combining such an approach with the EAT, while possible in principle, is hindered by the dependence of the cutoff parameter on the observed statistics, which has to be taken into account when defining a min-tradeoff function. Namely, the correction term, which is non-affine in the cutoff parameter would have to be included in the min-tradeoff function, and the constraints of the optimisation from which we obtain our min-tradeoff function would contain non-affine terms in the cutoff parameter, greatly complicating the derivation of an affine min-tradeoff function.
We leave for future work a complete analysis of how to adapt this framework, and overcome these limitations.

Besides, while presented for a specific protocol consisting of four coherent states, our security proof can be adapted to any other constellation of coherent states. It deserves further investigation to study how the key rate changes when using more states and whether and how one can approximate the rate of protocols using Gaussian modulation. It is in fact expected that, as it happens for four coherent states (see Fig~\ref{fig:Nc_Test}), the rates obtained when using coarse-grained probabilities will be very close to those obtained when using moments~\cite{lin2019asymptotic}. Moreover, our approach provides a wide framework for the security analysis of CVQKD since the EAT is naturally device independent. Thus, we build a finite size security proof from the asymptotic regime via a min-tradeoff function without making any assumptions on the attack besides the Markov condition \eqref{eq:Markov}, here trivially satisfied. This is a fact of importance, provided that the optimal attack for DM CVQKD is not known.

Finally, it is worth noting that, for excess noise $\xi\geq0.01$, both our proof and the one in~\cite{matsuura2021finite,yamano2022finite,Matsuura2023} require block sizes of the order of $10^{11}-10^{12}$ to obtain a positive key rate, which are significantly larger than what needed for DVQKD. It is an interesting open question to understand whether a different proof strategy (e.g. the GEAT) can improve this substantially or if this is an intrinsic requirement of CVQKD. Comparing our results with \cite{Matsuura2023}, for $\xi=0.01$ and $n=10^{12}$, our method provides higher rates for distances around $10$ km, although different modelling of the error correction efficiency \cite{leverrier2023information} complicates a direct comparison of the rates. 

To conclude, we provide a security proof for CVQKD protocols in which the information in parameter estimation consists of coarse-grained probabilities instead of moments, as done so far. The analysis consist of two main ingredients: (i) the application of EAT including a local tomography process to derive the finite-key rates (ii) the computation of the asymptotic key rates using the formalism of~\cite{lin2019asymptotic} for increasing number of photons. Our work therefore shows that use use of coarse-grained probabilities in parameter estimation opens new avenues to prove the security of discrete modulated CVQKD protocols, as well-established methods developed in DVQKD can be applied in a rather straightforward way without any significant impact on the obtained key rates.

\medskip

\acknowledgements{
We would like to thank Rotem Arnon-Friedman, Ian George, Shouvik Ghorai, Min-Hsiu Hsieh, Florian Kanitschar, Anthony Leverrier, Rotem Liss, Bill Munro, Gelo Noel Tabia, Enky Oudot, Stefano Pironio, Toshihiko Sasaki, Ernest Tan, Thomas Van Himbeeck and Shin-Ichiro Yamano for insightful discussions. We would also like to thank two anonymous referees, for conferences QCrypt and QIP, for their insightful comments. This work is supported by the ERC (AdG CERQUTE, grant agreement No. 834266, and StG AlgoQIP, grant agreement No. 851716), the Government of Spain (FUNQIP, NextGen Funds and Severo Ochoa CEX2019-000910-S), Fundaci\'{o} Cellex, Fundaci\'{o} Mir-Puig, Generalitat de Catalunya (CERCA and the postdoctoral fellowship programme Beatriu de Pin\'{o}s), 
the AXA Chair in Quantum Information Science, European Union's Horizon 2020 research and innovation programme under grant agreements No. 820466  (project CiviQ), No. 101114043 (project QSNP), No. 101017733 (project Veriqtas) within the QuantERA II Programme, and No. 801370 (2019 BP 00097) within the  Marie Sklodowska-Curie Programme.
}

\newpage
\begin{appendix}
\begin{widetext}
\section{Proof of Lemma \ref{Lemma:DimReduct}}\label{App:A}

Let $\mc{H}_{\hat{E}}$ be a Hilbert space. We begin with a pure state $\ket{\rho}_{AB\hat{E}}\in\mc{H}_{AB\hat{E}}$. Alice and Bob's measurements are then performed coherently by means of a series of isometries. Alice's measurement, as well as the random number generator determining what the round is used for, are described by

\begin{equation} \label{eq:WA}
    W^{\text{A}}_{ARX \leftarrow A} = \sum_{r=0}^2\sqrt{p_r} \left( \sum_{x_r} \sqrt{P^{x_r}_A}\otimes \ket{r}_{R}\otimes \ket{x_r}_{X} \right),
\end{equation}
where $p_0=p^\mathrm{key}, p_1=p^\mathrm{PE},p_2=p^\mathrm{tom}$, and  $P_A^{x_r}$ denotes the $x$-th POVM element applied by Alice when her random bit provides an outcome $r$
\begin{align*}
    \{P_A^{x_0}\}_{x_0} &= \{P_A^{x_1}\}_{x_1} = \{\ket{x}\bra{x}_A\}_{x=0}^3, \\
    \{P_A^{x_2}\}_{x_2} &= \{ \Gamma_A^x \}_{x=0}^{15}.
\end{align*}

Note that the POVM elements are given by a square root to preserve the isometric characteristics of the measurement. Register $R$ will announce whether the bit will be used for the generation of the key, parameter estimation or tomography, whereas $X$ stores the result of Alice's measurement. Bob will perform a heterodyne measurement, which he will later discretise according to the goal of the round. Such measurement is given by the isometry
\begin{equation}
    W^{\text{B}}_{BY\leftarrow B} = \int d^2 y \sqrt{\frac{\ket{y}\bra{y}_B}{\pi}} \otimes \ket{y}_{Y}.
\end{equation}
Where the integral is given by the fact that coherent states form a continuous basis. The classical communication of $R$ between Alice and Bob, which is wiretapped by Eve, can be expressed coherently by adding ancillary subsystems followed by CNOTs. 

\begin{equation}
V_{ [R]\leftarrow R}^\text{c1}=U^\text{CNOT}_{R:R'R''}\ket{00}_{R'R''},
\end{equation}
where $R'$ is distributed to Bob and $R''$ to Eve, and we have introduced the simplifying notation $[R]:=RR'R''$. Furthermore, $U^\text{CNOT}_{R:R'R''}$ is the unitary describing a double CNOT taking $R$ as control and $R'$ and $R''$ as targets. Now that Bob and Alice have made their public announcements, we have to apply a new isometry where Bob discretises the key,
\begin{align}
    V^{\text{K}}_{RY\hat{Z} \leftarrow RY} &= \ket{0}\bra{0}_R \otimes \sum_{z=0}^3 \sqrt{R_{Y}^z} \otimes \ket{z}_{\hat{Z}}   + \big(\ket{1}\bra{1}_R + \ket{2}\bra{2}_R \big)\otimes \mathbb{1}_{Y} \otimes \ket{\perp}_{\hat{Z}}.
\end{align}


Here, the set $\{R_Y^z\}_{z=0}^3$ represents the region operators for the discretisation in key rounds, whose definitions are given in \eqref{eq:R0}. 
%
%
The state that results after applying all the isometries is then given by (where we have omitted identities on systems not involved)

\begin{equation}\label{def:omega}
\ket{\omega}_{ABXY\hat{Z}[R]\hat{E}}=V^K  V^\text{c1}W^\text{B}W^\text{A}\ket{\rho}_{AB\hat{E}}.
\end{equation}

Finally, the key register is dephased by a pinching map $\mc{Z}':\hat{Z}\to\hat{Z}$, defined with the Kraus operators
\begin{align}
&Z_j=\pro{j}_{\hat{Z}} \otimes \mathbb{1},
\end{align}
for $j\in\{0,1,2,3,\perp\}$. Note that this is same definition as in \eqref{def:Z}, albeit here with the symbol $\perp$ included. We can now apply Theorem 1 from \cite{coles2012unification}, to show that 
\begin{align}\label{eq:Coles}
H(\hat{Z}|R''\hat{E})_{\mc{Z}(\omega)}=D(\omega_{ABXY\hat{Z}RR'}||\mc{Z}'(\omega_{ABXY\hat{Z}RR'}))
\end{align}
Now, the r.h.s. does no longer depend on $\hat{E}$. Let us also observe that in the marginal $\omega_{ABXY\hat{Z}RR'[P]}$ registers $RR'$ have decohered due to traceout of $R''$. We can then reformulate it as
\begin{align} \label{eq:decohered}
&\omega_{ABXY\hat{Z}RR'}=p^\text{key} \pro{00}_{RR'}\otimes\omega_{ABXY\hat{Z}}^\text{key}+(1-p^\text{key})\(\pro{11}_{RR'}+\pro{22}_{RR'}\) \otimes \omega_{ABXY}^\perp \otimes\pro{\perp}_{\hat{Z}},
\end{align}
where $p^\text{key}$ denotes the probability that Alice will use the round for the generation of the key. The state \eqref{eq:decohered} has a cq structure, so that by the properties of the relative entropy on cq-states \cite{wilde2013quantum}, we can simplify \eqref{eq:Coles} by splitting the state according to the classical registers $RR'$. Moreover, the state $\omega^{\perp}_{ABXY}\otimes\pro{\perp}_{\hat{Z}}$ is invariant under the pinching, so that the relative entropy for such state is zero. The whole process adds up to the equality

\begin{align} 
&D(\omega_{ABXYRR'\hat{Z}}||\mc{Z}'(\omega_{ABXYRR'\hat{Z}}))=p^\text{key}D(\omega^\text{key}_{ABXY\hat{Z}}||\mc{Z}(\omega^\text{key}_{ABXY\hat{Z}})),\label{eq:A13}
\end{align}
where we have also substituted $\mc{Z}'$ for $\mc{Z}$ since we have removed $\perp$ from the key register $\hat{Z}$. The explicit form of the key state $\omega_{ABXY\hat{Z}}^\text{key}$ is then given by
\begin{align}
\omega_{ABXY\hat{Z}}^\text{key}&=\frac{1}{p^\text{key}}\Tr_{R''\hat{E}}\left[\bra{00}_{RR'}\omega_{ABXY[R]\hat{Z}\hat{E}}\ket{00}_{RR'}\right].
\end{align}
Following the arguments provided in Appendix A of \cite{lin2019asymptotic}, we can further simplify (\ref{eq:A13}). First of all, the reduction of the state \eqref{def:omega} according to the properties of the relative entropy for cq-states has suppressed the sum over $r$ at \eqref{eq:WA}, leaving only the term related to the key generation, namely $r=0$. Hence, Alice's operator for the key state is given by
\begin{equation}
    W'^{\text{A}}_{AX \leftarrow A} = \sum_{x=0}^3 \pro{x}_A \otimes \ket{x}_X,
\end{equation}
where we used the fact that Alice's POVM elements in $A$ are projectors, so that the square root can be removed. This operator now merely copies and projects the information stored in $A$ to the new register $X$, which effectively represents an isometry that is invariant under the pinching (since both registers are not related to the key register $\hat{Z}$). Hence we can simplify this isometry by removing $X$, and the final operator for Alice will be a mere identity in $A$.

As for the key rounds, Bob only needs to obtain the discretised key variable $\hat{Z}$. Thus, he can group the POVM elements corresponding to a particular value of $\hat{Z}$, forming a \emph{coarse grained} POVM $\{R^i_B\}_{i=0}^3$ that acts directly on register $B$, and is given by the region operators defined in \eqref{eq:R0}. Hence, register $Y$ is not necessary and Bob's measurement and discretisation will thus be given by

\begin{equation}
W'^{\text{B}}_{B\hat{Z} \leftarrow B}=\sum_{z=0}^3\sqrt{R^z_B}\otimes\ket{z}_{\hat{Z}}.
\end{equation}
Now, the simplified maps for Alice and Bob are combined to provide the CP map $\mc{G}:AB\to AB\hat{Z}$  that represents the postprocessing, which as shown in \eqref{def:G} is given by the superoperator
\begin{equation}
    G = W'^{\text{A}} \otimes W'^{\text{B}} =  \mathbb{1}_A \otimes \sum_{z=0}^3 \sqrt{R_{B}^z} \otimes \ket{z}_{\hat{Z}}.
\end{equation}
We can now conclude with the redefinition of the relative entropy at \eqref{eq:A13} in terms of the postprocessing map,

\begin{equation}\label{eq:A21}
D(\omega^\text{key}_{ABXY\hat{Z}}||\mc{Z}(\omega^\text{key}_{ABXY\hat{Z}}))=D(\mc{G}(\rho_{AB})||\mc{Z}(\mc{G}(\rho_{AB}))).
\end{equation}

By definition, register $R''$ is identical to register $S$ in (\ref{eq:MinTradeOff1}) and for any $\mc{H}_{\hat E}$ and any $\ket{\rho}_{AB\hat{E}}\in{\mc{H}_{AB\hat{E}}}$ it holds $\mc{Z}(\omega)_{\hat{Z}S\hat{E}}=\(\id_{\hat{E}}\otimes\mc{M}^\text{EAT}(\rho)\)_{\hat{Z}S\hat{E}}$, for $\omega$ defined as in (\ref{def:omega}). Combining the equations (\ref{eq:Coles},\ref{eq:A13},\ref{eq:A21}), we obtain that for all $\tilde{p}\in\mc{P}_{\tilde{\mc{C}}}$,
\begin{align}
g(\tilde p)&=\inf_{\rho\in\Sigma(\tilde p)}H(\hat{Z}|R''\hat{E})_{\mc{Z}(\omega)}=p^\text{key}\inf_{\rho\in\Sigma(\tilde p)}D(\mc{G}(\rho_{AB})||\mc{Z}(\mc{G}(\rho_{AB}))),
\end{align}
which finishes the proof.

\section{Proof of Lemma \ref{lemma:conv}}\label{App:conv}
Let $p,q\in\mc{P}_{\tilde{\mc{C}}}$, $0\leq\lambda\leq1$. Without loss of generality we can assume that $\Sigma(p)$ is not empty. Then there exist states ${\rho}_{AB}\in\Sigma( p)$ and ${\tau}_{AB}\in\Sigma(q)$ such that 
\begin{align}
&p^\mathrm{key}D(\mc{G}(\rho_{AB})||\mc{Z}(\mc{G}(\rho_{AB})))=g(p),\\
&p^\mathrm{key}D(\mc{G}(\tau_{AB})||\mc{Z}(\mc{G}(\tau_{AB})))=g(q).
\end{align}
Let us now consider the flag state
\begin{equation}
\omega_{ABF}=\lambda\rho_{AB}\otimes\pro{0}_F+(1-\lambda)\tau_{AB}\otimes\pro{1}_F.
\end{equation}
It then holds \cite{wilde2013quantum}
\begin{align}
p^\mathrm{key}D(\mc{G}\otimes\id_{F}(\omega_{ABF})||\mc{Z}\circ\mc{G}\otimes\id_{F}(\omega_{ABF}))&=\lambda p^\mathrm{key}D(\mc{G}(\rho_{AB})||\mc{Z}(\mc{G}(\rho_{AB})))\nonumber \\
&\quad +(1-\lambda)p^\mathrm{key}D(\mc{G}(\tau_{AB})||\mc{Z}(\mc{G}(\tau_{AB})))\nonumber\\
&=\lambda g(p)+ (1-\lambda) g(q),\label{eq:43a}
\end{align}
As tracing out the flag system $F$ cannot increase the relative entropy, it holds
\begin{align}
&D(\mc{G}(\omega_{AB})||\mc{Z}(\mc{G}(\omega_{AB})))\leq D(\mc{G}\otimes\id_{F}(\omega_{ABF})||\mc{Z}\circ\mc{G}\otimes\id_{F}(\omega_{ABF})).\label{eq:44a}
\end{align}
Let now $c\in\tilde{\mc{C}}$. It then holds
\begin{align}
\bra{c}\Tr_{OS}\left[\mc{M}^\text{EAT,test}(\omega_{AB})\right]\ket{c}&=\lambda\bra{c}\Tr_{OS}\left[\mc{M}^\text{EAT,test}(\rho_{AB})\right]\ket{c}+(1-\lambda)\bra{c}\Tr_{OS}\left[\mc{M}^\text{EAT,test}(\tau_{AB})\right]\ket{c}\nonumber\\
&=\lambda p(c)+(1-\lambda) q(c).
\end{align}
This implies that ${\omega}_{AB}\in\Sigma\(\lambda p+(1-\lambda) q\)$. By definition of $g$, and eqs. (\ref{eq:44a}) and (\ref{eq:43a}), it then holds
\begin{align}
g(\lambda p+(1-\lambda) q)&\leq p^\mathrm{key} D(\mc{G}(\omega_{AB})||\mc{Z}(\mc{G}(\omega_{AB})))\nonumber\\
&\leq p^\mathrm{key} D(\mc{G}\otimes\id_{F}(\omega_{ABF})||\mc{Z}\circ\mc{G}\otimes\id_{F}(\omega_{ABF}))\nonumber\\
&= \lambda g(p)+ (1-\lambda) g(q),
\end{align}
finishing the proof.


\section{Upper bounding the classical smooth max entropy}\label{App:B}

Let $n\in\mathbb{N}$, and for $1=1,...,n$, let $Y_i$ be a binary classical random variable such that $P_{Y_i}(1)=p$ and $P_{Y_i}(0)=1-p$. Further, define classical random variable $X_i$ such that $X_i=\perp$ if $Y_i=0$.
Otherwise the values are chosen from an alphabet $\mathcal{X}$ such that $|\mathcal{X}\cup\{\perp\}|=d$.  We use the operator representation to describe the joint state as
\begin{align}
\rho_{X_1^nY_1^n}&=\sum_{x_1,...,x_n\in \mathcal{X}\cup\{\perp\}}\sum_{y_1,...,y_n=0}^1P_{X_1^nY_1^n}(x_1,...,x_n,y_1,...,y_n) \nonumber \\
&\quad \times \pro{x_1,...,x_n}_{X_1^n}\otimes\pro{y_1,...,y_n}_{Y_1^n},\label{eq:state}
\end{align}
etc. 
\begin{lemma}
For any $\epsilon>0$ it holds
\begin{equation}
H^\epsilon_{\max}(X_1^n|Y_1^n)_\rho\leq np\log{d} +\sqrt{\frac{n}{2}\ln{\frac{2}{\epsilon^{2}}}}\log{d}.
\end{equation}
\end{lemma}

\begin{proof}




Let $\epsilon>0$ and define $\delta:=\left(\frac{\ln2-{2}\ln\epsilon}{2n}\right)^\frac{1}{2}$. We can divide the sum in eq. (\ref{eq:state}) into a part with up to $\lfloor n(p+\delta)\rfloor$ terms with $Y_i=1$, hence non-trivial $X_i$, and a part with more than $\lfloor n(p+\delta)\rfloor$ such terms, $\rho_{X_1^nY_1^n} = \rho'_{X_1^n Y_1^n} +  \rho''_{X_1^n Y_1^n}$, where
\begin{align}
\rho'_{X_1^n Y_1^n} &= \sum_{x_1,...,x_n\in \mathcal{X}\cup\{\perp\}}\sum_{\substack{y_1,...,y_n \in \{0,1\}^n \\ \sum_{i} y_i \leq \lfloor n(p+\delta) \rfloor }} P_{X_1^nY_1^n}(x_1,...,x_n,y_1,...,y_n) \nonumber \\
&\quad \times \pro{x_1,...,x_n}_{X_1^n}\otimes\pro{y_1,...,y_n}_{Y_1^n},\\
\rho''_{X_1^n Y_1^n} &= \sum_{x_1,...,x_n\in \mathcal{X}\cup\{\perp\}}\sum_{\substack{y_1,...,y_n \in \{0,1\}^n \\ \sum_{i} y_i > \lfloor n(p+\delta) \rfloor }} P_{X_1^nY_1^n}(x_1,...,x_n,y_1,...,y_n)\nonumber \\
&\quad \times \pro{x_1,...,x_n}_{X_1^n}\otimes\pro{y_1,...,y_n}_{Y_1^n},
\end{align}

Let us define
\begin{align}
\kappa&:=\Tr\left[\rho''_{X_1^nY_1^n}\right]=\sum_{k=\lfloor n(p+\delta)\rfloor+1}^np^k(1-p)^{n-k}\binom{n}{k}, 
\end{align}
and note that by Hoeffding's inequality, it holds $\kappa\leq e^{-2n\delta^2}\leq\frac{\epsilon^2}{2}$.
By \cite{tomamichel2015quantum}, Lemma 3.17, it then holds for the purified distance
\begin{equation}
    P(\rho_{X_1^nY_1^n},\rho'_{X_1^nY_1^n})\leq\sqrt{\left\|\rho_{X_1^nY_1^n}-\rho'_{X_1^nY_1^n}\right\|_1+\Tr\(\rho_{X_1^nY_1^n}-\rho'_{X_1^nY_1^n}\)}=\sqrt{2\kappa}\leq\epsilon
 \end{equation}
hence $\rho'$ is in the $\epsilon$-ball around $\rho$. Consequently, as only the non trivial $X_i$ contribute to the max entropy, it holds
\begin{align}
H^\epsilon_{\max}(X_1^n|Y_1^n)_\rho
&\leq H_{\max}(X_1^n|Y_1^n)_{\rho'}\label{eq:11a}
\leq \log d^{\lfloor n(p+\delta)\rfloor}.
\end{align}
Inserting our choice for $\delta$ completes the proof.

\end{proof}

Now, let's add conditioning on an event $\Omega$ that occurs with probability $p_\Omega>0$.  We can express the state (\ref{eq:state}) as $\rho_{X_1^nY_1^n}=\Pr[\Omega] \rho_{X_1^nY_1^n}|_\Omega+(1-\Pr[\Omega])\rho_{X_1^nY_1^n}|_{\lnot\Omega}$.

\begin{lemma}\label{lem:Omar}
For any $\epsilon>0$ and $0<p_\Omega\leq1$ it holds
\begin{equation}
H^{\epsilon}_{\max}(X_1^n|Y_1^n)_{\rho|_\Omega}\leq np\log{d} +\sqrt{\frac{n}{2}\ln{\frac{2}{\epsilon^2 \Pr[\Omega]}}}\log{d}.
\end{equation}
\end{lemma}

\begin{proof}



Let $\epsilon>0$ and define $\delta:=\(\frac{\ln2-\ln p_\Omega-2\ln\epsilon}{2n}\)^\frac{1}{2}$. Again, we divide $\rho_{X_1^nY_1^n}|_\Omega = \rho'_{X_1^n Y_1^n}|_\Omega +  \rho''_{X_1^n Y_1^n}|_\Omega$, where
\begin{align}
\rho'_{X_1^n Y_1^n}|_\Omega &= \sum_{x_1,...,x_n\in \mathcal{X}\cup\{\perp\}}\sum_{\substack{y_1,...,y_n \in \{0,1\}^n \\ \sum_{i} y_i \leq \lfloor n(p+\delta) \rfloor }} P_{X_1^nY_1^n}(x_1,...,x_n,y_1,...,y_n|\Omega)\nonumber \\
&\quad\times \pro{x_1,...,x_n}_{X_1^n}\otimes\pro{y_1,...,y_n}_{Y_1^n},\\
\rho''_{X_1^n Y_1^n}|_\Omega &= \sum_{x_1,...,x_n\in \mathcal{X}\cup\{\perp\}}\sum_{\substack{y_1,...,y_n \in \{0,1\}^n \\ \sum_{i} y_i > \lfloor n(p+\delta) \rfloor }} P_{X_1^nY_1^n}(x_1,...,x_n,y_1,...,y_n|\Omega)\nonumber \\
&\quad\times\pro{x_1,...,x_n}_{X_1^n}\otimes\pro{y_1,...,y_n}_{Y_1^n},
\end{align}


Let us define

\begin{align}
\kappa:=&\Tr\left[\rho''_{X_1^nY_1^n}|_\Omega\right]\\
&=\sum_{k=\lfloor n(p+\delta)\rfloor+1}^n \Pr\(|\{i:Y_i=1\}|=k|\Omega\)\\
&=\frac{1}{\Pr[\Omega]}\sum_{k=\lfloor n(p+\delta)\rfloor+1}^n \Pr\(|\{i:Y_i=1\}|=k\cap\Omega\)\\
&\leq\frac{1}{\Pr[\Omega]}\sum_{k=\lfloor n(p+\delta)\rfloor+1}^np^k(1-p)^{n-k}\binom{n}{k}\\
&=\frac{1}{\Pr[\Omega]} \Pr[k> n(p+\delta)].
\end{align}
By Hoeffding's inequality, it holds $\kappa\leq\frac{e^{-2n\delta^2}}{p_\Omega} \leq\frac{\epsilon^2}{2}$. 
By \cite{tomamichel2015quantum}, Lemma 3.17, it then holds for the purified distance
\begin{equation}
    \Pr(\rho_{X_1^nY_1^n}|_\Omega,\rho'_{X_1^nY_1^n}|_\Omega)\leq\sqrt{\left\|\rho_{X_1^nY_1^n}|_\Omega-\rho'_{X_1^nY_1^n}|_\Omega\right\|_1+\Tr\(\rho_{X_1^nY_1^n}|_\Omega-\rho'_{X_1^nY_1^n}|_\Omega\)}=\sqrt{2\kappa}\leq\epsilon
 \end{equation}
hence $\rho'|_\Omega$ is in the $\epsilon$-ball around $\rho|_\Omega$. Consequently, as only the non-trivial $X_i$ contribute to the max entropy, it holds
\begin{align}
H^\epsilon_{\max}(X_1^n|Y_1^n)_{\rho|_\Omega}
&\leq H_{\max}(X_1^n|Y_1^n)_{\rho'|_\Omega}
\leq \log d^{\lfloor n(p+\delta)\rfloor}.
\end{align}
Inserting our choice for $\delta$ completes the proof.




\end{proof}

\section{Perturbative analysis for finite-key distillation} \label{App:D}

The framework presented in \eqref{sec:MinTrade} provides a method to derive finite secret key rates via Frank-Wolfe and the dual problem \eqref{eq:SDP3a} using a min-tradeoff function. This technique has the drawback that  the dual variables appear in the correction terms of the finite key rate \eqref{eq:FiniteKeyRate}, and thus these corrections are not directly optimised. Eventually, this poses a problem that harms the performance of our methodology, especially for small numbers of rounds $n$.  We overcome this obstacle by employing a perturbative analysis via genetic algorithms---for the original dual objective \eqref{eq:SDP3a}, we apply a modification

\begin{equation} \label{eq:PerturbedDual}
    \ell^0_{\tilde p_0,\varepsilon'}(\vec{\nu}, \vec{\mu}) \rightarrow \ell^0_{\tilde p_0,\varepsilon'}(\vec{\nu},\vec{\mu}) -  \zeta_{\tilde p_0}(\vec{\kappa})
\end{equation}
given by a perturbative term
\begin{equation}
    \zeta_{\tilde p_0}(\vec{\kappa}) = \kappa_0 \lVert \tilde{p}_0 \rVert_1 + \kappa_1 \lVert \tilde{p}_0 \rVert_2 + \kappa_2 \lVert \tilde{p}_0 \rVert_\infty.
\end{equation}
Let us observe that, for $\vec{\kappa}\in \mathbb{R}^3_+$, the perturbed dual objective serves as a lower bound for the original one. Moreover, the perturbation acts as a term that reduces the spread of the dual variables when the new objective dual is employed, and since both maximisations are executed under the same set of constraints, we can solve the SDP given by the perturbed objective function, and insert the solution in the original version \eqref{eq:SDP3a} to build the min-tradeoff function. In order to derive useful values for $\vec{\kappa}$ that balance a minimised value for the dual variables with an increased performance in the finite-key analysis, we use a step inspired by genetic algorithms. The method goes as follows.

\begin{algorithm}{Genetic subroutine}
    
    \vspace{0.3cm}
    
    For round $m \in \{1,...,5\}$ perform the following steps:
    \begin{enumerate}
        \item Generate a set of random vectors $\{\vec{\kappa}_j\}_{j=1}^{100}$, with coefficients ranging between $0$ and $10^{-3}$.
        \item For every vector, define a perturbed objective dual as \eqref{eq:PerturbedDual} for the maximisation \eqref{eq:SDP3a}. 
        \item Solve the resulting SDP, and use the solution to calculate both the min-tradeoff function and the finite secret key rate according to Theorem \ref{FinteSizePhys}.
        \item Record as $S_m$ the highest finite key rate achieved for the iteration.
        \item Discard the vectors that provide finite rates below the percentile 10, and those whose value for the spread of the min-tradeoff function is above the percentile 95.
        
        \item Combine randomly the remaining vectors in 90 pairs $\left(\vec{\kappa}_i,\vec{\kappa}_j\right)$, with the associated finite rates $(F_i,F_j)$, and create a new population of vectors. 
        \item For each pair $\left(\vec{\kappa}_i,\vec{\kappa}_j\right)$, create a vector $\vec{\kappa}_k$ by means of genetic crossings. For $l \in \{0,1,2\}$, every entry $\kappa_{l,k}$ of the new vector is evaluated with the following procedure:
        \begin{itemize}
            \item Generate a random value $p \in [0,1]$ using a uniform distribution. If $p>0.9$, assign to $\kappa^l_k$ a random value between $10^{-8}$ and $10^{-1}$.
            \item  Otherwise, draw a value from a binomial distribution with a bias $F_i/(F_i+F_j)$ towards the zero. If the value is zero, evaluate $\kappa_{l,k} := \kappa_{l,i}$, and otherwise $\kappa_{l,k} := \kappa_{l,j}$.
        \end{itemize}
        \item Complete the new population by adding the 10 best performing vectors (in terms of finite key rates) from the previous round.
        \item Start over the routine with the new set of vectors.
    \end{enumerate}
\end{algorithm}\label{Alg:Genetic}
Once this subroutine is complete, the final secret key rate is given by the maximum value in $\{S_1,...,S_{5}\}$. We note that this approach is purely heuristic, and it can be further improved by finely adjusting the number of iterations and the range of values for the coefficients.
Nevertheless, it provides a proper framework to derive reliable, high key rates in the finite setting. In particular, it enables us to reduce the value of the coefficients $\epsilon$, $\epsilon^\mathrm{c}_\mathrm{PE}$ and $\epsilon^\mathrm{tom}$ without affecting noticeably the results of our method.

\end{widetext}

\end{appendix}

\bibliographystyle{quantum}
\bibliography{CV}

\end{document}